\documentclass{llncs}

\pdfoutput=1 

\newcommand{\Lang}{HASL}

\usepackage{amsmath,amssymb,dsfont}

\newcommand{\haddad}{linear hybrid }

\newcommand{\cosmos}{\mbox{\textup{C}\scalebox{0.75}{{\textsc{OSMOS}}}}}

\newcommand{\uppaal}{\mbox{\textup{U}\scalebox{0.75}{{\textsc{PPAAL}}}}}
\newcommand{\uppaalsmc}{\mbox{\textup{U}\scalebox{0.75}{{\textsc{PPAAL-SMC}}}}}
\newcommand{\plasma}{\mbox{\textup{P}\scalebox{0.75}{{\textsc{LASMA}}}}}

\newcommand{\prism}{\mbox{\textup{P}\scalebox{0.75}{{\textsc{RISM}}}}}

\newcommand{\apmc}{\mbox{\textup{A}\scalebox{0.75}{{\textsc{PMC}}}}}
\newcommand{\ymer}{\mbox{\textup{Y}\scalebox{0.75}{{\textsc{MER}}}}}
\newcommand{\mrmc}{\mbox{\textup{M}\scalebox{0.75}{{\textsc{RMC}}}}}

\newcommand{\vesta}{\mbox{\textup{V}\scalebox{0.75}{{\textsc{ESTA}}}}}

\newcommand{\const}{\mathsf{Const}}
\newcommand{\lconst}{\mathsf{lConst}}
\newcommand{\flow}{\textit{flow}}
\newcommand{\updates}{\mathsf{Up}}

\newcommand{\false}{\mbox{\tt false}}

\newcommand{\aconst}{\alpha}

\newcommand{\IN}{\!\in\!}
\newcommand{\TIMES}{\!\times\!}

\newcommand{\ignore}[1]{}

\usepackage{geometry}
\usepackage{lipsum}

\usepackage{graphicx}
\usepackage{times}
\usepackage{tikz}
\usetikzlibrary{arrows,shapes,snakes,automata,backgrounds,petri}
\usepackage[latin1]{inputenc}
\usepackage[noend]{algorithmic}
\usepackage{algorithm}

\usepackage{epsfig}
\usepackage{subfigure}
\usepackage{multirow}

\usepackage{tikz}
\usepackage{algorithm}

\title{Analysing oscillatory trends of discrete-state stochastic processes through HASL statistical model checking}
\author{Paolo Ballarini\inst{}}

\institute{%
Ecole Centrale Paris, France\\
\email{paolo.ballarini@ecp.fr}
}

\date{}

\begin{document}

\maketitle

\begin{abstract}
The application of formal methods to the analysis of stochastic oscillators  has
been at the focus of several research works in recent times. In this paper we
provide insights on the application of an expressive temporal logic formalism,
namely the Hybrid Automata Stochastic Logic (HASL), to that issue. 
We show how one can take advantage of  the expressive power of the HASL logic to
define and assess relevant  characteristics 
of (stochastic) oscillators. 
\end{abstract}

\section{Introduction}
\label{sec:intro}

Oscillations are a relevant type of dynamics which characterises the behaviour of several types of system in different domains, notably 
in the field of  biological modelling. 

The analysis of oscillations is a well established subject in applied mathematics for which different approaches 
exist. For example for systems described in terms of Ordinary Differential Equations (ODEs)  limit-cycle analysis can be used to assess 
oscillation characteristics (e.g. period and amplitude of oscillations). On the other hand  
signal processing methods such as, for example, Fast  Fourier Transformation (FFT)  or autocorrelation analysis 
can be used to extract the oscillatory characteristic of   a given signal, i.e. a sequence of points resulting from the observed system (being  
it the actual system under investigation or a model representing it). 

In recent times the study of  oscillatory systems has attracted the attention 
of research in the area discrete-state stochastic models  (in the remainder we will  refer to these kind of oscillatory models simply as stochastic oscillators)  yielding to a number of research works aimed at the application of temporal logic reasoning to characterisation of oscillations~\cite{Ballarini20102019,Spieler13,DBLP:journals/tcsb/AndreiC12,DBLP:conf/isola/DavidLLMPS12}.
The goal in that respect is, quite simply, to adapt (stochastic) model checking techniques so that, given a model $M$, 
one is capable to obtain answers to questions such as: \emph{``does $M$ oscillates?''} , \emph{``where the peaks of  oscillations are located  ?''} 
\emph{``what is the (average) period of oscillations?''}. 
Since here we refer to stochastic models  answering  such questions usually boils down to assessing some  distribution of probability 
(e.g. assessing the steady-state distribution of   $M$, and/or the PDF of the period duration, and/or the PDF of the location of the peaks 
of oscillations). 
So far analysis of oscillations  through stochastic model checking   have been mainly obtained through application 
of the Continuous Stochastic Logic (CSL)~\cite{BHHK03}  (in some cases joint to its reward-based extensions~\cite{KNP07a}) approach 
or similarly expressive (linear-time) variants (e.g. Metric Interval Temporal Logic~\cite{DBLP:conf/isola/DavidLLMPS12}). 
Interestingly, in recent times, Spieler~\cite{Spieler13} has shown how  \emph{qualitative}, such as  \emph{``does a model $M$  oscillates sustainably?"}, 
as well as \emph{quantitative}, such as, for example, \emph{ ``what is the period  of oscillation of a sustained  oscillator $M$?"},  queries 
 can be formally assessed (in CSL form) 
by coupling a continuous-time Markov chain (CTMC) model with a timed automaton (TA) ``monitor" 
capable of  identifying noisy-periodic traces.

In this paper we extends Spieler's approach by considering   a recently introduced 
formalism, i.e.   the Hybrid Automata Stochastic Logic (HASL)~\cite{BDD+11}, as 
a means for   studying  stochastic oscillators.  

\paragraph{Paper contribution.} 
The paper main contribution is one of demonstrating the effectiveness of the HASL formalism as 
a means to effectively  specify and automatically estimate oscillation related measures. 
We consider two different approaches: the first concerned with 
assessing the oscillation period, the second  concerned with measuring  the oscillation 
peaks (hence  the oscillation amplitude). 
We define  two specific  types of  linear hybrid automata (LHA)  
that, when synchronised with an oscillatory stochastic process, 
are  capable of detecting the periods, respectively the peaks, of its trajectories 
and compute \emph{on-the-fly} classical characteristics like the average duration 
or the average amplitude of the oscillations as well as more sophisticated ones 
like the period fluctuation, which allow for assessing the \emph{regularity} 
of an oscillator.

We demonstrate such contributions by considering a well-known case study. i.e. the analysis of 
a model of the circadian clock~\cite{VKBL02}.

\paragraph{Paper organisation.} In Section~\ref{sec:hasl} we introduce the 
HASL formalism. In Section~\ref{sec:oscillationsHASL} we describe the basic contribution of the paper, namely 
the application of HASL to the analysis of oscillations. In Section~\ref{sec:circlock} we demonstrate the 
HASL-based analysis of oscillations on an example of  biological oscillator. 
We wrap up the paper with some concluding remarks in Section~\ref{sec:conc}.

\section{HASL  model checking}
\label{sec:hasl}

HASL framework belongs to the family of so-called statistical model checking methods,    whose goal is  
to produce estimates of a (formally specified) target measure  through  
sampling  of  the  trajectories 
of the model\footnote{as opposed to  numerical stochastic model checking which requires the complete construction of a model's state-space to assess the exact value of the target measure.}.
HASL is an automata-based type of logic, meaning that it employs  automata, and specifically  linear hybrid automata (LHA) as 
machineries for characterising the properties to be investigated.  
This yields the  main feature of HASL, that is,  its expressive power  which in this paper we 
are going to demonstrate in respect to the oscillation analysis problem. 

Simply speaking the HASL model checking procedure works as follows: given a model ${\cal D}$ and a certain dynamics of interest encoded by an LHA ${\cal A}$, 
the HASL  model checker  
 samples trajectories of the synchronised process ${\cal D}\!\times\! {\cal A}$, hence selecting 
only those paths of ${\cal  D}$ that are \emph{accepted} by ${\cal A}$ and using them for 
estimating the confidence-interval of a given target measure (in the following denoted $Z$), 
a quantity defined as a function of the LHA variables. 

In the following we recall the basics formal elements for HASL:  the characterisation of Discrete Event Stochastic Process (DESP)
i.e the pertaining class of models, the characterisation of LHA and of the corresponding synchronised process 
${\cal D}\!\times\!{\cal A}$, and that of target measure $Z$. For more details we refer the reader to~\cite{BDD+11}.

\subsection{Discrete Event Stochastic Processes}


We refer to a DESP  as a discrete-state stochastic process  consisting of an  
enumerable  set of states 
and whose dynamic is  triggered by a  set of (time consuming) discrete events. 
We do not consider any restriction on the nature 
of the distribution associated with events\footnote{hence, in essence, a DESP corresponds to a generalised semi-Markov processes~\cite{Gly83,ACD91}}. 
Otherwise said a DESP is a family of random variables $\{X(t)\mid t\in\mathbb{R}_{\geq 0}\}$ representing time and where, in the context of this paper,  
$\mathbb{N}^n$ is assumed to be the support of $X$ (i.e. we talk in this case of an $n$-dimensional DESP population model, 
see Definition~\ref{def:populationmodel}). 
Below we formally define the components a DESP consists of. Such characterisation is 
useful to provide an algorithmic formulation of the dynamics of a DESP (see below)  which is at the basis 
of the HASL statistical model checking procedure. 
\paragraph{Notation} For $A$ a generic set we denote 
$dist(A)$  the set of possible probability distributions whose support is $A$, that is, 
$dist(A)=\{\mu:\Sigma_A\to\mathds{R}^+ | (A,\Sigma_A,\mu)\}$ where  
$(A,\Sigma_A,\mu)$ is a probability space. Observe that 
depending on the nature of $A$  the 
corresponding probability distributions $\mu\in dist(A)$ are either \emph{continuous} (if $A$ is dense) 
or \emph{discrete} (if $A$ is finite/discrete). 

\begin{definition} [DESP]
A DESP is a tuple\\ 
${\cal D}=\langle S, \pi_0 , E, Ind, enabled, delay, choice, target\rangle$ where
\begin{itemize}\itemsep.1pt
	\item $S$ is an enumerable (possibly infinite) set of states, 
	\item $\pi_0 \in dist(S)$ is the initial distribution on states,
	\item $E$ is a finite set of events, 
	\item $Ind$ is a set of functions from $S$ to  ${\mathds R}$ called state indicators
        (including the constant functions),
	\item $enabled\!: S \rightarrow 2^E$ are the enabled events in each state with for all $s \in S$, $enabled(s)\neq \emptyset$.
	\item
	 $delay\!: S \times E \rightarrow dist(\mathds{R}^+)$ is a partial function
        defined for pairs $(s,e)$ such that $s \in S$ and $e \in enabled(s)$.
        \item 
        $choice\!: S \times 2^E \times \mathds{R}^+\rightarrow dist(E) $ is a partial  function  defined for tuples $(s,E',d)$ such that $E' \subseteq enabled(s)$
        and such that the possible outcomes of the corresponding distribution are restricted to $e \in E'$.
	\item $target\!: S \!\times\! E \!\times \mathds{R}^+ \!\to\!S$ is a partial function describing state changes through events
        defined for tuples $(s,e,d)$ such that $e \IN enabled(s)$.
\end{itemize}
\label{def:desp}
\end{definition}

A  \emph{configuration} of a DESP consists of a triple $(s\mathbin{,}\tau\mathbin{,} sched)$ with $s$ being the current state, $\tau\IN{\mathds R}^+$ the current time  and $sched : E \rightarrow {\mathds R}^+ \cup \{+\infty\}$ being the function that describes the occurrence time of each scheduled event  
($+\infty$ if an event is not yet scheduled).
Observe that a scheduler $sched$ essentially describes the state of  events' queue in a given configuration 
of the DESP, thus all (currently) enabled events will have a finite scheduled time $sched(e)\!<\!\infty$, whereas 
non-enabled events will be associated with an infinite delay, that is,  $sched(e)\!=\!\infty$. 
Thus, within the algorithm for generating a trajectory of a DESP,  a scheduler $sched$ 
provides the occurrence time of   the  next event to occur (see also Algorithm~\ref{alg_desp}). 
In the remainder we denote $Conf\!=\! S\TIMES \mathds{R}^+\TIMES Sched$ the set of possible configurations of a DESP (where $Sched$ denotes the set of possible schedules functions for the events of the DESP). 
Also for a configuration  $c\!=\!(s\mathbin{,}\tau\mathbin{,} sched)\IN Conf$, we denote $c(s)$, $c(\tau)$ and $c(sched)$  the 
state $s$, respectively the time $\tau$ and the schedule $sched$ of  configuration $c$. 

For a state $s$, $enabled(s)$ is the set of events enabled in $s$. 
For  $e \IN enabled(s)$, $delay(s,e)$ 
is the distribution of the delay between the enabling of $e$ and its possible occurrence. 
Furthermore, if we denote  $\delta_{m}$ the  delay of the earliest event in the current configuration 
 $(s\mathbin{,}\tau\mathbin{,} sched)$  of the 
process, and $E_{min} \!\subseteq\! enabled(s)$ 
the set of events with earliest delay, then $choice(s,E_{min},\delta_m)$ describes how
the conflict between the concurrent events in $E_{min}$ is randomly resolved: i.e.  
$choice(s,E_{min},\delta_m)(e')$ is the probability that $e'\IN E_{min}$ will be selected
hence occurring with delay  $\delta_{m}$. Finally function
$target(s,e,d)$ denotes the target state reached from $s$ on
occurrence of $e$ after waiting for  $d$ time units.

\paragraph{Dynamics of a DESP.}
The evolution of  a DESP   ${\cal D}$ 
can be informally summarised by an iterative procedure 
consisting of the following steps (assuming  $(s\mathbin{,}\tau\mathbin{,}sched)$  is the current configuration  of ${\cal D}$): 
1) determine the set $E_{min}$ of events enabled in state $s$ and  with minimal delay $\delta_m$; 
2)   select the next event  to occur  $e_{next}\IN E_{min}$ by  resolving  
conflicts (if any) between concurrent  events through  probabilistic 
choice according to $choice(s,E_{min},\tau)$; 3) determine  the new configuration of 
the process resulting from the occurrence of $e_{next}$, this in turns consists of 
three sub-steps: 3a)  determine the new state 
resulting from occurrence of $e_{next}$, i.e. $s'=target(s,e_{next},\delta_m)$; 
3b) update the current time to account for the delay of occurrence of $e_{next}$, i.e. $\tau=\tau+\delta_m$; 
3c) update the schedule of events according to the newly entered state $s'$ (this implies 
setting the schedule of no longer enabled events to $+\infty$ as well as determining 
the schedule of newly enabled events by sampling through the corresponding 
distribution).
  Such procedure is (semi-formally) summarised in   Algorithm~\ref{alg_desp}. 
\begin{algorithm}                      
\caption{Evolution of a DESP}          
\label{alg_desp}                           
\begin{algorithmic}                    
    \STATE initial\_configuration: $(s\mathbin{,}\tau\mathbin{,}sched)$
       \WHILE{$Enabled(s)\neq\emptyset$}
       \STATE  $E_{min}\!=\!min(\cup_{e\in enabled(s)} sched(e))$ 
       \STATE  $e_{next} = choice(s,E_{min},\tau)$
       \STATE $\delta_m = sched(e_{next})$
        \STATE  $s' = target(s,e_{next},\delta_m)$
        \STATE  $\tau' = \tau + \delta_m $
         \STATE   $sched(e) \!= +\infty$ \ ($\forall e\!\not\in\! enabled(s'))$
          \STATE   $sched(e) \!=\! sample(delay(s',e))$\  ($\forall e\!\in\! enabled(s'))$
    \ENDWHILE
\end{algorithmic}
\end{algorithm}
 
\noindent
A path (or trajectory) of a DESP is a sequence of configurations  $\sigma\!=\! c_1,c_2,c_3,\ldots$ 
resulting from the execution of the procedure highlighted by Algorithm~\ref{alg_desp}.  
We formalise this in the following definition. 
The notion of DESP path  will be used  later on for reasoning about the dynamics of a DESP and in particular for reasoning about oscillations. 
%
 \begin{definition}[Path of a DESP]
For a   DESP ${\cal D}=\langle S, \pi_0 , E, Ind, enabled, delay, choice, target\rangle$ with $Conf$ the set of its configurations 
we define the set of finite paths as $Path^*\subseteq\bigcup_{n\IN\mathds{N}} Conf^n$. 
We denote 
$\sigma\!=\!(c_0,c_1\ldots, c_n)\IN Path^*$, 
where $\pi_0(c_0(s))\!>0$ and $\forall 0\leq i< n$, 
$\exists e\IN E_{min}(c_i(s))$ such that $c_{i+1}(s)=target(c_i(s),e,c_i(\tau))$. 
By extension we denote $Path^{\omega}$ as the set of infinite path and 
$Path\!=\! Path^*\!\cup\! Path^\omega$ as the of all paths of a DESP. 
\end{definition}
\noindent 
In the remainder we might refer to  a DESP path using $\sigma\!=\!(c_0,c_1\ldots, c_n)$ 
or, depending on the context,  simply indicating the corresponding sequence of  states  $\sigma\!=\!(c_0(s),c_1(s)\ldots, c_n(s))$, 
or simply $\sigma\!=\!(s_0,s_1\ldots, s_n)$. 
Furthermore  for  a path $\sigma\!=\!(c_0,c_1\ldots, c_n)$ 
we  use the following notations: 
for $i\!\in\!\mathds{N}$, $\sigma[i]=c_i(s)$ denotes the $i$-th state, while for $t\!\in\!\mathds{R}^+$,    
$\sigma @ t$ denotes the state in which $\sigma$ is at time $t$, 
that  is, $\sigma @ t\!=\! \sigma[i]$ such that $i$ is the smallest $i$ with 
$t\!\leq\! c_i(\tau)$

Since in this paper we deal with the analysis of discrete-state biological models representing the 
evolution of the molecular population of $n$ species, 
we introduce the notion of DESP population model. 
 \begin{definition}[DESP Population Model]
 \label{def:populationmodel}
A DESP model for $n\IN\mathds{N}$ population types is a DESP  
${\cal D}=\langle S, \pi_0 , E, Ind, enabled, delay, choice, target\rangle$ with $S\!\subseteq\!\mathds{N}^n$.
\end{definition}

 \begin{definition}[DESP Observed Species]
For ${\cal D}$ a DESP population model  with $n$ species we define ${\cal D}_i$ the 
observed $i^{th}$ process, with $1\!\leq\! i\!\leq\! n$,   as the process resulting from 
${\cal D}$  by observing only the $i^{th}$ component of each state of ${\cal D}$. 
Thus each $s\!=\! (s_1,\ldots, s_i,\ldots s_n)\IN S$ of ${\cal D}$ corresponds to state 
$s_i\IN S_i$ of ${\cal D}_i$. 
\end{definition}
\noindent
By extension for $\sigma\IN Path$ a path of a DESP population model we denote $\sigma_i$ 
the   $i^{th}$ projection of $\sigma$, thus if 
$\sigma\!=\! (s^1_0,\ldots s^n_0), (s^1_1\ldots s^n_1), \ldots $ then 
$\sigma_i\!=\! s^i_0 ,s^i_1\ldots  $



\paragraph{Indicator functions.} 
In the definition of  DESP  we include a set of \emph{indicator functions} denoted $Ind$. 
An indicator  $\alpha\!\in\! Ind$ maps  states of a DESP to real values $\alpha:S\to{\mathds R}$. 
DESP indicators describe what information can be seen 
by an LHA during the  synchronisation  with a DESP. 
Specifically, indicators  appear in  various parts of a synchronising LHA (see Definition~\ref{def:dta}): 
in the \emph{location invariants}  (function $\Lambda$), in a location's \emph{flow}, and 
in the  \emph{edge constraints} ($\const$ and $\lconst$, within $\to$) 
 and \emph{edge updates} ($\updates$) of an LHA edge. 
We denote $Prop\!\subseteq\! Ind$  the subset of boolean valued  indicators called \emph{propositions}, i.e., 
for $\alpha^*\!\in\! Prop$, $\alpha^*:S\to\{0,1\}$. 
Indicators are evaluated against states. Thus for $s\!\in\! S$, and $\alpha\IN Ind$ an indicator, $\alpha(s)$ 
denotes the value of $\alpha_i$ in state $s$. 
Specific details about how   indicators are applied within the characterisation of an LHA 
are given in Section~\ref{sec:HASL}.



\paragraph{DESP in terms of GSPN}
For implementation convenience, in the context of HASL and in particular of the associated model checking tool \cosmos~\cite{cosmos,BDDHP-qest11}, 
we represent DESP models in terms of stochastic petri nets, and more precisely we adopt (the non-markovian extension\footnote{GSPN with timed transitions 
associated to generic probability distributions, that is, not necessarily Negative Exponential as with the standard GSPN definition~\cite{ABC+95}.} 
of) Generalised Stochastic Petri Net (GSPN)~\cite{ABC+95} as the high-level input formalism 
for expressing a DESP model. Thus, in this context,  DESP indicators 
are actually GSPN indicators, that is: 
they are expressions which contain  references to the (marking of the) places of a GSPN model. 
For the sake of brevity here we assume familiarity with the GSPN formalism, 
referring the reader to the literature~\cite{ABC+95} for details. 
GSPN semantics is briefly presented later on through description of a simple 
GSPN model (see Figure~\ref{fig:lha2}). 

\paragraph{Example:  DESP indicators within  LHA.} 
In the LHA of Figure~\ref{fig:lha2} (right) the indicator {\fontfamily{phv}\selectfont protA}, 
which refers to the marking of the GSPN  place named {\fontfamily{phv}\selectfont protA} in 
Figure~\ref{fig:lha2} (left),  is used within the updates of the self-loop edges of location $l_0$. 
Specifically  indicator {\fontfamily{phv}\selectfont protA} 
is used to update the LHA variable  $a$  with the  current number of tokens contained in GSPN place 
{\fontfamily{phv}\selectfont protA}. 
Similarly in the LHA  of Figure~\ref{fig:lhaoscill} the GSPN place indicator $A$ 
(which refers to the GSPN place named $A$ of the GSPN in Figure~\ref{fig:circClock}) 
is used within the invariant constraints $A\!\leq\! L$, $L\!\leq\! A\!\leq\! H$ and $A\!\geq\! H$ 
associated respectively with locations $low$, $mid$ and $high$ (where $L,H\IN\mathds{R}$ are just symbolic names for   
two real-valued constants used for representing a generic version of the ${\cal A}_{per}$ LHA: 
in practice concrete instances of ${\cal A}_{per}$ are obtained by actual instances of $L,H$, e.g., $L\!=\!1$ and $H\!=\! 10$). Such invariants essentially state  that entering the locations $low$, $mid$ and $high$ 
depend on the current marking of place $A$ (see Section~\ref{sec:periodHASL} for more details). 
\subsection{Hybrid Automata Stochastic Logic}
\label{sec:HASL}

\ignore{
The use of statistical methods instead of numerical ones gives us the
possibility to relieve the limitations in terms of model and
properties that were inherent to numerical methods. When numerical
model checking was focusing on markovian models, statistical methods
permit to use a very wide range of distributions, and to synchronise
such a model with a non decidable class of automata with linearly
evoluting variables, complex updates and guards. We also consider a
more complex evaluation of systems. We are no more limited to the
probability with which a property is satisfied, we can also compute
the expected value of performability  parameters such as waiting time,
number of clients in a system, production cost of an item.
}

The  Hybrid Automata Stochastic Logic, introduced in~\cite{BDD+11}, 
extends  Deterministic Timed Automata (DTA) logics  for describing properties of Markov chain models~\cite{DHS09,CHKM09}, 
by employing  LHA (a generalisation of  DTA) 
as  instruments for characterising specific dynamics of an observed DESP model. 
An HASL formula consists of two elements: 1)  
a so-called   synchronising  LHA, i.e.  an LHA  enriched with (state and/or event) \emph{indicators} 
of the observed  DESP and 2) a target expression (see grammar~(\ref{a1})) which  expresses the quantity to be evaluated. 
 The synchronised LHA  is used for selecting the trajectories 
that correspond to the behaviour to of interest. The target expression 
indicates what statistics, i.e. what function of the synchronised LHA data variables,  
will be assessed with respect to  the trajectories selected by the LHA. 

In the following we formally introduce the notion of synchronised LHA  and then informally 
describe the stochastic process resulting from the product of a DESP and a synchronised LHA. 

\begin{definition}
\label{def:dta}
 A \emph{synchronised \haddad automaton} 
 is a tuple 
 $\mathcal A\!=\!\langle E, L, \Lambda, I, F, X, \flow, \rightarrow \rangle$
 where: 

\begin{itemize}
\item
 $E$ is a finite alphabet of events;
 \item
 $L$ is a finite set of locations;
  \item
  $\Lambda: L \rightarrow Prop$ is a location labelling function;
   \item
$I\subseteq L$ is the initial locations;
    \item
$F \subseteq L$ is the final locations;
   \item
$X=(x_1,...x_n)$ is a $n$-tuple of data variables; 
    \item $\flow : L \mapsto Ind^n$ associates an $n$-tuple of indicators with each location 
    (the  $i^{th}$ projection $\flow_i$ denotes the flow of change of   variable $x_i$).
    \item 
    $\rightarrow \subseteq L \times \left( (2^E  \times \const ) \uplus    (\{\sharp\}  \times  \lconst) \right)   \times \updates \times L$
  is the set of edges of the LHA ,
 \end{itemize}
where $\uplus$ denotes  the disjoint union, $\const$ and $\lconst$ denotes the set of possible \emph{constraints}, 
respectively \emph{left closed constraints}, associated with ${\cal A}$ (see details below),  
 $\updates$ is the set of possible updates for the variables of ${\cal A}$ 
 and $Prop\!\subseteq\! Ind$ denotes the subset of boolean valued 
  DESP indicators. 
 \end{definition}

Before presenting informally the synchronisation of a DESP with an LHA we 
start by   describing the various parts of an LHA. 
In what follows we denote indicators  symbolically by greek letters $\alpha,\alpha'\IN Ind$, 
while we use capital letters  $A,B,\ldots$ to refer to names of GSPN places (within 
concrete indicators instances) and $x_1,x_2,\ldots $ to denote LHA variables. Thus, for example, $\alpha\!\equiv\! A+2B$ is an 
indicator whose value is given by the sum of the marking of place $A$ with twice the marking of 
place $B$. 

\paragraph{Location proposition:} function $\Lambda$ associates  each location $l\IN L$ with a  \emph{proposition} (also called  \emph{location invariant} in the remainder)  $\Lambda(l)\IN Prop$
representing a condition under which a location can be entered. A location proposition 
consists of a boolean combination of inequalities involving  DESP indicators and 
has the following form   $\Lambda(l)\!\equiv\! \bigwedge_i (\!\alpha_i \!\prec\! \alpha'_i\!)$ 
with $\alpha_i,\alpha'_i\!\in\! Ind$,  and $\prec \in\!\{=, <, >, \leq,\geq\}$. Notice that 
indicators can be constant functions, thus, for example, a location proposition
 may consist of comparing  indicators' values against  constant thresholds, as in, e.g., $\Lambda(l)\!\equiv\!A\!\!\geq\! 10$, 
or it may consist of comparing different  indicators  one another, as in, e.g., $\Lambda(l)\!\equiv\! A\!\leq\! B$ or 
$\Lambda(l)\!\equiv\! A\!\leq\! B^2 \sqrt{C}$. 
A location proposition (given by $\Lambda$) is  shown by a label 
next to the location it refers to.
For convenience no label is shown next to unconstrained   locations, i.e., locations associated to a  tautology like $\top\equiv (\alpha_i \!=\! \alpha_i)$. 
Location propositions are evaluated against states of a DESP. Thus for $s\IN S$ a  DESP state 
and $l\IN L$ a location of an LHA we say that $s$ satisfies the invariant $\Lambda(l)$, 
denoted $s\!\models\!\Lambda(l)$,  if $\Lambda(l)(s)\!=\! {\bf true}$ 
(where $\Lambda(l)$ is the value of the boolean expression obtained by 
replacing each indicator $\alpha\IN\Lambda(l)$ with its value $\alpha(s)$). 
Furthermore given two edge locations $l$ and $l'$ 
we say that the their respective invariants are \emph{inconsistent},  
denoted $ \Lambda(l)\land\Lambda(l') \Leftrightarrow \false$,  if there cannot 
exist a state $s$ that satisfies $\Lambda(l)\land\Lambda(l')$.  
For example,  if $\Lambda(l)\!\equiv\! A\!\leq\! 2$ and $\Lambda(l')\!\equiv\! A\!>\! 2$ 
 then trivially $ A\!\leq\! 2 \land A\! >\! 2\Leftrightarrow \false$.  
This means that $l$ and $l'$ are  mutually exclusive, which is a necessary 
condition for LHA with multiple initial locations  (see conditions  {\bf c1} below). 

\paragraph{Edge constraint:} edge constraints  describe necessary conditions for an edge to 
be traversed. We denote $\const$ (resp. $\lconst$)  the set of constraints (resp. left-closed constraints) of an LHA edge. 
An \emph{edge constraint} consists of a boolean combination of inequalities involving both DESP indicators and 
LHA variables. They have  the following form   $\gamma\!\equiv\! \bigwedge_j (\sum_{1 \leq i \leq n}\aconst_{ij} x_i \! \prec\! \alpha'_j\!)$ 
with $\alpha_{ij},\alpha'_j\!\in\! Ind$, $x_i\IN X$ and $\prec \in\!\{=, <, >, \leq,\geq\}$. 
Simple examples of edge constraints can be: $\gamma\!\equiv (2x_1 \!+\! 3 x_2 \!\leq\! 5)$ or also $\gamma\!\equiv (A x_1  \!=\! 5)$. 
Given a location $l$ of the LHA and a state $s$ of the DESP, 
the inequalities $\gamma_j\!\equiv\! \sum_{1 \leq i \leq n}\aconst_{ij} x_i  \! \prec\! \alpha'_j$ a constraint $\gamma$ consists of    evolve 
linearly   with time hence each inequality  gives an interval of time during which the constraint $\gamma$ is satisfied. 
We say that a constraint is left closed if, whatever the current state $s$  
(defining the values of the DESP indicators), the time 
at which the constraint is satisfied is a union of left closed intervals 
(for example, $\gamma\!\equiv (x_1 \!\geq\! 5)$ is  left-closed 
whereas $\gamma\!\equiv (x_1 \! >\! 5)$ is not). 
We denote $\lconst\!\subseteq\!\const$ the subset of left-closed constraint. 
For efficiency the constraint of autonomous-edges (see below) must be left-closed. 

\noindent
Edge constraints are evaluated against pairs $(s,\nu)\IN S\!\times\! Val$ where $s\IN S$ is a  
state of a DESP and $\nu:X\rightarrow \mathds{R}\IN Val$ is a  \emph{valuation}  
that  maps every LHA data variable to a real value (we denote $Val$ the set of all possible valuations).
For $\nu\IN Val$, $\nu(x)$ denotes the value of variable $x$ through valuation $\nu$. 
Given $\gamma_j\!\equiv\! \sum_{1 \leq i \leq n}\aconst_{ij} x_i  \! \prec\! \alpha_j$ 
an inequality contained in an edge-constraint $\gamma\!\equiv\! \bigwedge_j \gamma_j$, 
its interpretation w.r.t. $\nu$ and $s$, denoted $\gamma_j(s,\nu)$, is defined by $\gamma_j(s,\nu)= \sum_{1 \leq i \leq
n}\aconst_{ij}(s) \nu(x_i)\!\prec\!\alpha'_j(s)$. 
We write $(s,\nu) \models \gamma_j$ if
$\gamma_j(s,\nu) \!=\! {\bf true}$ and, by extension,  
$(s,\nu) \models \gamma$ iff  $(s,\nu) \models \gamma_j$ for all $j$. 
Furthermore given two edge constraints $\gamma$ and $\gamma'$ 
we say that their conjunction $\gamma\land \gamma'$  is \emph{inconsistent}, 
denoted $\gamma\land \gamma' \Leftrightarrow \false$,  if there 
exists no combination $(s,\nu)\IN S\!\times\! Val$ that satisfies it. 
For example,  if $\gamma\!\equiv\! x_1\!\leq\! 2$ and $\gamma'\!\equiv\! x_2\!>\! 2$ 
are the constraints for two edges then trivially $( x_1\!\leq\! 2)\land (x_1\! >\! 2)\Leftrightarrow \false$, 
meaning the two edges cannot be concurrently enabled (see conditions  {\bf c2} and  {\bf c3} below).

 \paragraph{Edges update:}  an edge update $U\!=\!(u_1, ...,u_n)\IN Up$ is an $n$-tuple of functions 
 characterising how each LHA variable $x_k$ is going to be updated on traversal of the edge. 
 Each function  $u_k$ ($1\!\leq\! k\!\leq\! n$)  of an edge update $U\!=\!(u_1, ...,u_n)\IN Up$ is of the form 
$x_k=\sum_{1 \leq i \leq n}\aconst_i x_i+c$ where the $\aconst_i$ and $c$ are DESP indicators. 
Similarly to edge constraints, updates are evaluated against pairs $(s,\nu)\IN S\!\times\! Val$.  Given an update $U=(u_1,\ldots,u_n)$,
we denote by $U(s,\nu)$ the valuation defined by
$U(s,\nu)(x_k)=u_k(s,\nu)$ for $1 \leq k \leq n$. \\

 \paragraph{Locations flow:}  a location flow is an $n$-tuple  of indicators $flow(l)\!=\!(\alpha_1,\ldots, \alpha_n)$, 
where $\alpha_i\IN Ind$  describes the gradient at which  variable $x_i\IN X$ changes while the automaton sojourns  
in    location $l$. 
Specifically when location $l$ is entered the rate of change of each $x_i$ is established by 
the valuation, w.r.t. to the  state  the DESP is at  on entering of $l$, of the corresponding $\alpha_i$. 
Observe that, if each $\alpha_i$ in $flow(l)$ is a constant function (e.g. $\alpha_i=c_i$, with $c_i\IN\mathds{R}$) 
then  each variable $x_i$ changes at constant rate throughout the sojourn in $l$. 
However this is not necessarily the case for variables whose flow is given by a non-constant indicator, 
like, for example, $\alpha_i\!=\! c_1A + c_2B$, with $c_1,c_2\IN\mathds{R}$ and $A,B$ representing the 
marking of two GSPN places named $A$ and $B$. In this case the flow of change of $x_i$ depends on the 
marking of places $A$ and $B$, and such marking may change during the sojourn in $l$, for example 
if a synchronising self-loop edge $l\to l$ exists which synchronises with some DESP event whose occurrence 
modify the marking of $A$ or $B$. \\

\noindent
Having described the DESP indicators dependent elements of an LHA we now 
see how they are all combined within the characterisation of an LHA edge.

\paragraph{Edges of  an LHA.}  
  An edge 
  $l\xrightarrow{E',\gamma,U}l'$ 
  of  an LHA is labelled  by: 1) a constraint $\gamma$, 2) a set of event labels $E'$, 3) an update $U$. 
   An edge  can be either \emph{synchronous} 
or \emph{autonomous}. 
A synchronous edge is one whose traversal is triggered by the occurrence of an event 
of the DESP  in particular an  event  $e\IN E'\!\subseteq E$ where $E'$ is the set of event names labeling the edge. 
An autonomous edge, on the other hand,  is one whose traversal is independent of the occurrence 
of DESP events, hence  the event label for autonomous edges is  $E'\!\equiv\!\sharp$, where  $\sharp$ 
is the label used for representing a ``pseudo-event''.

The class of LHA for HASL is further restrained by the following conditions: 

\begin{itemize}\itemsep.01pt
   
     \item {\bf c1 (\underline{initial determinism):}}
    $\forall  l \neq l' \IN I$,  $\Lambda(l) \wedge \Lambda(l') \Leftrightarrow \false$. 
    This means that independently  of the
    interpretation of the  indicators, hence of  the synchronising  DESP model, at most 
    one initial location $l \IN I$ can have its constraint   $\Lambda(l)$ verified.

    \item {\bf c2 (\underline{determinism on events:})}
    $\forall E_1,E_2 \!\subseteq\! E$ $:\! E_1 \!\cap\! E_2 \neq \emptyset,$ $\forall l,l',l'' \in L,$
    if $l'' \xrightarrow{E_1,\gamma,U}{l}$ and $l'' \xrightarrow{E_2,\gamma',U'}{l'}$ are two distinct transitions, then either
    $\Lambda(l) \wedge \Lambda(l') \Leftrightarrow \false$ or
    $\gamma \wedge \gamma' \Leftrightarrow \false$. Again this
    equivalence must hold whatever the interpretation  of the
    indicators occurring in $\Lambda(l)$, $\Lambda(l')$, $\gamma$ and
    $\gamma'$. 
    \item {\bf c3 (\underline{Determinism on $\sharp$:})}
    $\forall l,l',l'' \in L,$
    if $l'' \xrightarrow{\sharp,\gamma,U}{l}$ and  $l'' \xrightarrow{\sharp,\gamma',U'}{l'}$
    are two distinct transitions, then either
    $\Lambda(l) \wedge \Lambda(l') \Leftrightarrow \false$ or
    $\gamma \wedge \gamma' \Leftrightarrow \false$. 
    \item {\bf c4 (\underline{no $\sharp$-labelled loops:})}
  For all sequences \\
  $l_0 \xrightarrow{E_0,\gamma_0,U_0} l_1 \xrightarrow{E_1,\gamma_1,U_1}
  \cdots \xrightarrow{E_{n-1},\gamma_{n-1},U_{n-1}}{l_n}$ such that $l_0 = l_n$,
  there exists $i \leq n$ such that $E_i \neq \sharp$.
\end{itemize}

\paragraph{Synchronisation of LHA and DESP.}
The role of a synchronised LHA ${\cal A}$ is to select specific trajectories of 
 a corresponding DESP ${\cal D}$ while collecting relevant data (maintained in the LHA variables)  along the execution. 
Synchronisation is technically  achieved through the product process  ${\cal D}\times {\cal A}$ 
whose formal characterisation,  for the sake of brevity, we omit in this paper: we provide 
however an intuitive description of the ${\cal D}\times {\cal A}$ semantics. 

The product ${\cal D}\!\times\! {\cal A}$ is itself  a DESP whose states   are  triples $(s,l,\nu)$ where $s$ is the current  state of the ${\cal D}$, 
 $l$ the current location of the ${\cal A}$ and $\nu\!:\!X\!\to\mathbb{R}$  the current valuation of the  variables of ${\cal A}$. 
Formally the set of states of the product process ${\cal D}\!\times\! {\cal A}$ is defined as  $S'=(S\times L \times Val) \uplus \{\bot\} $, 
where $Val$ denotes the set of possible variables' valuations and 
$\bot$ denotes the \emph{rejecting} state, i.e., the state entered when synchronisation fails, hence when a trajectory is 
rejected (see below). 
 Notice that a configuration of the product DESP ${\cal D}\times {\cal A}$ has the following form $((s,l,\nu), \tau,sched')$, 
 where $(s,l,\nu)$ is the current state of ${\cal D}\times {\cal A}$, $\tau\IN{\mathds R}^+$ is the current time, and $sched'$ is the schedule of the 
enabled  events of ${\cal D}\times {\cal A}$. 
The synchronisation  starts from the initial state $(s,l,\nu)$, where $s$ an the initial state of the DESP (i.e. $\pi_0(s)>0$), 
$l$ is an initial location of the LHA (i.e. $l\IN I$) and the LHA variables are all initial set to zero (i.e. $\nu=0$)\footnote{Notice that because of the ``initial-nondeterminism'' 
of LHA there can be at most one initial state for the product process.}. 


From the initial state the synchronisation process evolves through  transitions 
where each transition   corresponds to  traversal of either 
a synchronised or an autonomous edge   of the LHA\footnote{notice that because of the  determinism constraints of the LHA edges 
(conditions {\bf c2} and {\bf c3})  
at most only one autonomous or synchronised edge can ever be enabled in any location of the LHA.}. Furthermore 
if an autonomous and a synchronised edge are concurrently enabled the
autonomous transition is taken first. 
Let us suppose that $(s,l,\nu)$ is the current state of process  ${\cal D}\!\times\! {\cal A}$ 
and describe how the synchronisation evolves. 
If in the current location of the LHA (i.e. location $l$ of the current state $(s,l,\nu)$)
there exists  an enabled autonomous
edge $l\xrightarrow{\sharp, \gamma,U} l'$, then that edge 
will be traversed leading to a new  state $(s,l',\nu')$ where the DESP state ($s$) 
is unchanged whereas  the  new location $l'$ 
and the new variables' valuation $\nu'$  might differ from  $l$, respectively  $\nu$, 
as a consequence of the  edge traversal. 
On the other hand if   an event $e$ of process ${\cal D}$ (corresponding to  transition $s\xrightarrow{e} s'$ of ${\cal D}$)  
occurs in  state $(s,l,\nu)$, either  an enabled  synchronous edge  $l\xrightarrow{E',\gamma,U} l'$ (with $e\IN E'$) exists leading to new state 
 $(s',l',\nu')$ of process  ${\cal D}\!\times\! {\cal A}$ (from which   synchronisation  will continue)  
 or  the synchronisation  halts hence the trace is rejected (formally this is achieved with  the system entering the rejecting state $\bot$). 


\paragraph{Enabling of an LHA edge.} 
Let us briefly describe how the enabling, hence the traversal, of an LHA edge is established. 
Let  $(s,l,\nu)$ be the current state of process  ${\cal D}\!\times\! {\cal A}$.  
An  edge $l\xrightarrow{E, \gamma,U} l'$ being it \emph{autonomous} or \emph{synchronous} originating in $l$ is enabled 
if the following two conditions hold: 1) if the edge constraint is satisfied in state $(s,l,\nu)$ (i.e., if $(s,\nu)\models\gamma$) 
2) if the location invariant of the target location $\Lambda(l')$ is satisfied in the state $s'$ reached by traversal of the edge, i.e., 
if $s'\!\models\!\Lambda(l')$ (observe that if the considered edge is autonomous then necessarily $s'\!=\!s$, whereas if it is synchronous 
then possibly $s'\!\neq s$). 
Finally for a synchronous edge to be enabled, in addition to 1) and 2),  it must be the case that 
the DESP event $e$ occurring while in  $(s,l,\nu)$ is  captured by the the edge, i.e., $e\IN E$.\\

\paragraph{Remarks.} 
The above described synchronisation of a DESP and an LHA, which  HASL model checking is based on, 
requires  certain properties to hold, namely: uniqueness, convergence and termination of the synchronisation. 
This means that for ${\cal A}$ a synchronised LHA then for any (infinite) path $\sigma$ of a synchronising DESP model: 
1) there must be exactly one synchronisation with ${\cal A}$, 2) synchronisation cannot 
go on indefinitely due to an infinity of consecutive autonomous events, 3) path $\sigma$ should lead to an 
\emph{absorbing} state (i.e. a final location of the ${\cal A}$ or the rejecting state $\bot$) with probability 1. 
The uniqueness property is  guaranteed by constraint {\bf c1}, {\bf c2} and {\bf c3} of the LHA definition whereas 
 convergence is a consequence of constraint {\bf c4}. 
On the other hand termination of the synchronisation is not explicitly guaranteed, 
however  can be ensured by structural properties of $\mathcal A$ and/or $\mathcal D$.

\ignore{
It should also be said that the restriction to linear equations in the constraints 
and to a linear evolution of data variables can be relaxed, as long as they are not involved in autonomous transitions. 
Polynomial evolution of constraints could easily be allowed for synchronised edges for which 
we would just need to evaluate the expression at a given time instant. 
Since the best algorithms solving polynomial equations operate in PSPACE~\cite{Can88}, 
such an extension for autonomous transitions cannot be considered for obvious efficiency reasons. 
}

\begin{table*}
{\bf Example (synchronisation of DESP and LHA)}: To understand how synchronisation of  a DESP with an LHA works 
let us consider  a simple example. 
Figure~\ref{fig:lha2} depicts a  toy DESP model in GSPN form (on the left) coupled  with a simple  
 LHA (on the right). The GSPN model represents the basic steps 
of gene expression: 1) binding/unbinding of an activator protein to  
the promoter of $gene\_A$;  2) of \emph{transcription} of a gene into an mRNA molecule; 
3) degradation of the mRNA 4) \emph{translation} of the mRNA into the expressed protein  $prot\_A$. 
The states of the DESP consist of 4-tuples $s\!=\!$ ({\fontfamily{phv}\selectfont protA}, {\fontfamily{phv}\selectfont geneA} ,{\fontfamily{phv}\selectfont A\_geneA}, {\fontfamily{phv}\selectfont mrnA}) $\!\in\!\mathds{N}^4$, corresponding to the marking of the 4 places of the GSPN, whereas 
the event set is $E\!=\!\{bind,unbind, degrade, transc, transl\}$, corresponding to the 5 timed-transitions of the GSPN. 
The LHA ${\cal A}$, on the other hand, consists of: two locations, $l_0$ (initial) and $l_1$ (final), and three data variables $t$ (a clock), $n$ (for counting 
the number of occurrences of the $transc$ event) and $a$ (for keeping track of the population of 
$prot\_A$). Notice that the \emph{invariant} of both locations is $\Lambda(l_0)\!=\!\Lambda(l_1)\!=\! \top$, 
(hence no label is associated to $l_0$, $l_1$), meaning that both locations can be entered without constraint. 
The initial state of the product process ${\cal D}\times {\cal A}$ is $s_0\!=\!((2,1,0,0),l_0,\nu_0)$, where $\nu_0$ is 
the zero valuation (i.e., $\nu_0(t)\!=\!\nu_0(n)\!=\! \nu_0(a)\!=\! 0$). 
${\cal A}$ has two  \emph{synchronised}  (self-loop) edges 
$l_0\xrightarrow{\{transc\},n\!< N,\{n++,a=\text{{\fontfamily{phv}\selectfont protA}}\}} l_0$, which synchronises with 
 occurrences of the $transc$ event, and $l_0\xrightarrow{E\setminus \{transc\},n\!< N,\{a=\text{{\fontfamily{phv}\selectfont protA}}\}} l_0$,
which synchronises with  occurrences of any other event but $transc$ (i.e., $E\setminus \{transc\}$). 
The \emph{constraint} for both synchronised edges is $n<N$ which means they can be traversed as long as the number of observed 
occurrences of $transc$, which is stored in $n$, is less than $N$. 
Both \emph{updates} for  the two synchronised edges refer to a single \emph{indicator}, namely {\fontfamily{phv}\selectfont protA}, 
whose value is given by the marking of the GSPN place labelled {\fontfamily{phv}\selectfont protA}, but they 
are slightly different. The update for the edge which synchronises with $transc$ is 
$\{n\!+\!+,a=\text{{\fontfamily{phv}\selectfont protA}}\}$, meaning that whenever the edge is traversed (i.e., on occurrence of a $transc$ event) the counter $n$ is 
incremented and the current marking of place {\fontfamily{phv}\selectfont protA} is stored in $a$. 
On the other hand the update for the edge which synchronises with $E\setminus \{transc\}$ the update is simply 
$\{a=\text{{\fontfamily{phv}\selectfont protA}}\}$ as clearly $n$ must be incremented only on occurrence of $transc$. 
Furthermore ${\cal A}$ has an \emph{autonomous} edge $l_0\xrightarrow{\sharp\!,\!(\boldsymbol {n\!=\! N})\!,\!\emptyset} l_1$ leading 
to the final location $l_1$. Such edge gets enabled as soon as its constraint $(n\!=\! N)$ is satisfied, that is, as soon as a state $s_N\!=\!((n_1,n_2,n_3,n_4),l_0,\nu_N)$,  
is reached with $\nu_N$ being a valuation such that $\nu_N(n)\!=\! N$. In any such state $s_N$ the autonomous edge is traversed 
(leading to state $s_{stop}\!=\!((n_1,n_2,n_3,n_4),l_1,\nu_N)$)
and the synchronisation stops.  
\end{table*}

\begin{center}
\begin{figure*}[ht]
\includegraphics[scale=0.4]{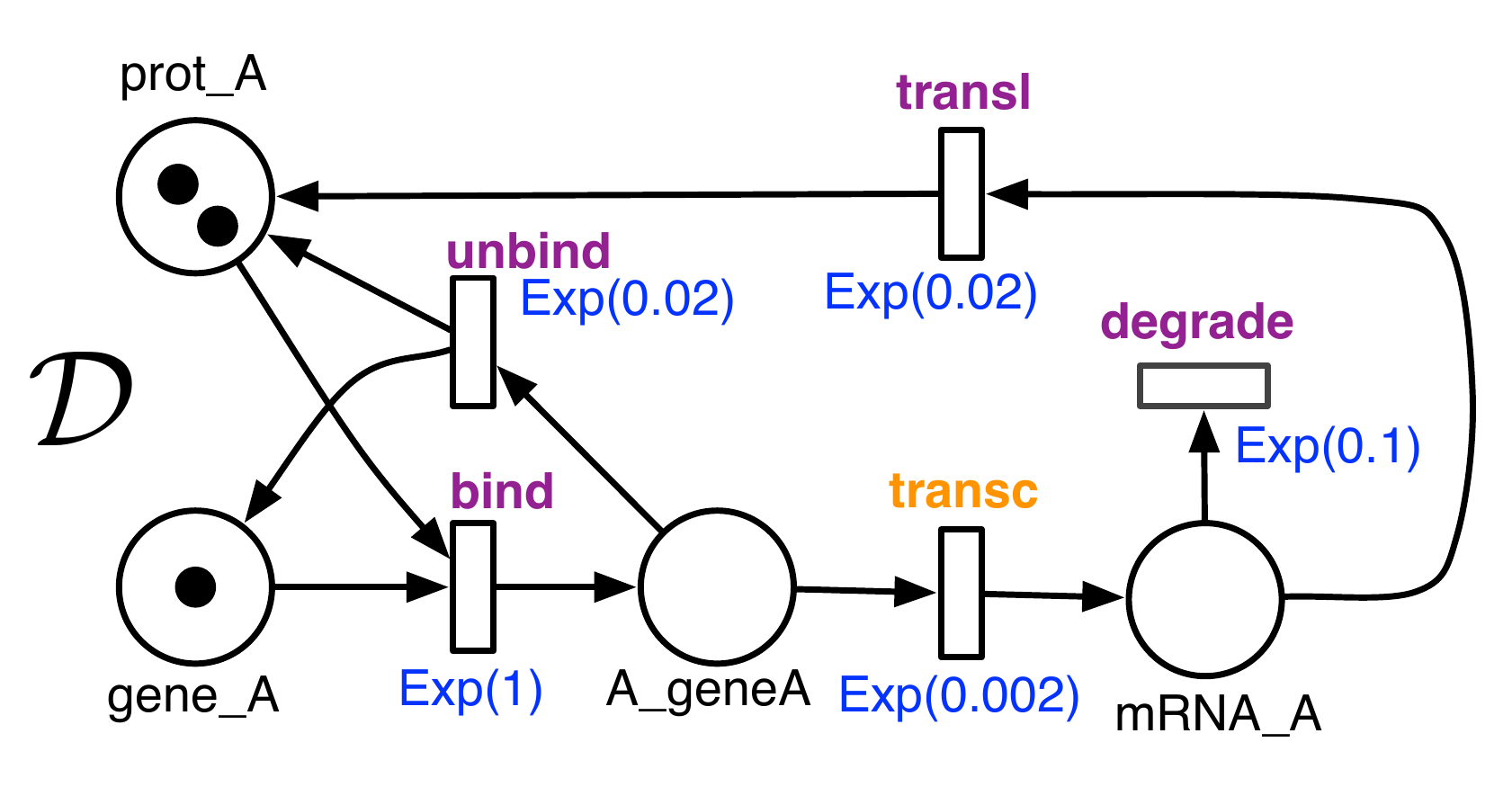}
\hspace{10ex}
\includegraphics[scale=0.4]{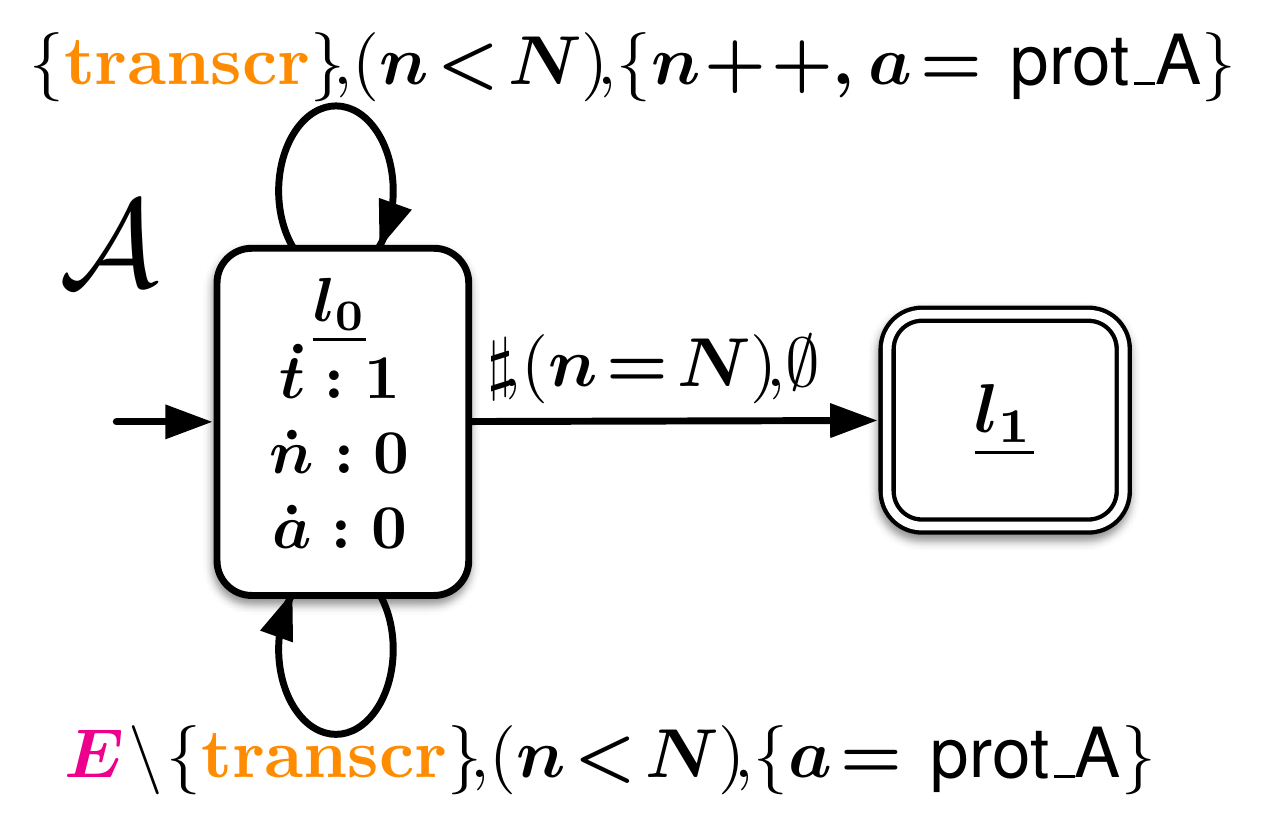}
\caption{Synchronisation between a DESP toy model (left) representing basic steps of gene expression and an LHA (right) which selects 
paths containing $N$ occurrences of the \emph{transcription} event}
\label{fig:lha2}
\end{figure*}
\end{center}

\paragraph{HASL expressions.}
The second component of an \Lang\ formula is an expression related to the automaton. 
Such an expression, denoted $Z$,  
is defined by a specific grammar~\cite{BDD+11}
of which here we consider  only the basic elements given in~(\ref{a1}). 

\begin{equation}
\label{a1} 
\begin{split}
 Z  ::= &\  E[Y]\ |\  P \\
 Y  ::= &\ last(y)\ |\ min(y)\  |\ max(y)\ |\ avg(y)\\ 
 y  ::= &\ c\ |\ x\ |\ y+y\ |\ y \times y\ |\ y/y
\end{split}
\end{equation}

\noindent

\noindent
$Z$ is either either an \emph{expectation} expression $E[Y]$,  or a \emph{probability} expression $P$. 
An expectation expression $Z=E[Y]$ represents the  expected value of  a  random variable $Y$ built on top 
of basic path operators ($last(y)$, $min(y)$, $max(y)$ , $avg(y)$). 
Each such path operator take as argument $y$ an algebraic combination  of the LHA data variables  $x$, 
and is evaluated  along a (synchronised) path that is accepted by the automaton. 
Intuitively the meaning of path operators is as follows: $last(y)$ represents the value    
that expression $y$ has at the instant a path is accepted, while $min(y)$ ($max(y)$, respectively $avg(y)$) 
represents the minimum (maximum, respectively average) value assumed by $y$ along an accepted path. 
Expression $Z\!=\!P$, on the other hand,  simply represents the probability that a  path  is accepted by the LHA. 
This is given by the ratio between the number of accepted paths and total number of paths generated 
throughout a simulation experiment. 

\noindent
In recent updates  the COSMOS model checker~\cite{BDDHP-qest11} has been enriched with facilities  for  assessing the Probability (Cumulative) Distribution Function (PDF, respectively CDF)   
of the value that an expression  $Y$ takes at the end of a synchronising path. 
Notice that PDF and CDF HASL expressions, are  only high-level macros supported by the COSMOS tool in order to   
give the user the possibility to  straightforwardly specify PDF/CDF measures\footnote{Otherwise 
PDF/CDF measures can be encoded explicitly in an LHA but such encoding would usually result 
in a rather complex LHA.}.  
Thus  COSMOS supports the 
following syntax for estimating a PDF measure: $Z=PDF(Y, s, l , h)$,  
where $Y$ is the path dependent expression whose PDF is to be estimated while $l$ and $h$ are the lower, respectively higher, 
bound of the interval representing the support of $Y$ (i.e.  estimation of the PDF of $Y$ is done assuming $Y$ takes value in $[l,y]$) 
and $s\!<\!(h-l)$ is the width of each sub-interval in which the considered support $[l,y]$ is discretised. 
Thus during estimation of  $Z=PDF(Y, s, l , h)$ COSMOS internally maintains 
a  counter for each of the  $(h-l)/s$ sub-intervals. Each such counter is incremented  
if  the value of $Y$ on acceptance of a trace falls in the corresponding sub-interval. 
Then the value returned by COSMOS for $Z=PDF(Y, s, l , h)$  is the array of frequencies 
obtained by dividing each of the above counters by the total number of generated trajectories.


\paragraph{Example}
Having introduced the HASL expression we can now consider some  examples of complete HASL 
formula referred to the model of Figure~\ref{fig:lha2}. 
\begin{itemize}
\item $ \phi_1\!\equiv\!({\cal A},E[last({t})])$:  representing the average  time for completing  $N$ \emph{transcriptions} 
\item ${\phi_2\!\equiv\!({\cal A}, E[max({a})])}$: representing the maximum population reached by protein A within the first $N$   \emph{transcriptions} 
\item ${\phi_3\!\equiv\!({\cal A}, PDF({last(t),0.1,0,10)})}$:  representing the PDF of the delay for  completing $N$  \emph{transcriptions} computed over 
the interval $[0,10]$ with a discretisation step of $0.1$ 

\end{itemize}

\noindent
Formulae $\phi_1,\phi_2,\phi_3$ refer all to the same LHA ${\cal A}$ (Figure~\ref{fig:lha2} right) which means 
the corresponding target measures are estimated with respect to the sampled trajectories 
of the same type (in this case containing exactly $N$ occurrences of the $transc$ event). 
 $\phi_1,\phi_2,\phi_3$  however differ in respect to 
the target expression $Z$. For $\phi_1$ the expression to be estimated is $Z_1\!=\!  E[last({t})]$, 
which is: the expected value that the LHA variable $t$ exhibits at the end ($last(t)$) of a synchronised trace. 
This means that for each trace $\sigma$ sampled from the process ${\cal D}\times {\cal A}$ the value that $t$ 
at the moment $\sigma$ is accepted (i.e., on occurrence of the $N$-th $transc$ event) is retained as a sample for the confidence-interval estimation of the mean value of 
$t$. 
For $\phi_2$ the expression to be estimated is $Z_2\!=\!  E[max({a})]$, which is: 
the expected value of the maximum ($max(a)$)  that  LHA variable $a$ exhibited along a synchronised trace. 
Observe that the maximum of an LHA variable along a trace is automatically computed on-the-fly during the sampling of a trace 
so that the value $max(a)$ for a sampled trace $\sigma$ is known straight away on acceptance of $\sigma$. 
Thus expression $Z_2\!=\!  E[max({a})]$  represents the expected value of the maximum number of protein A 
observed over sampled traces containing $N$ transcription events. 
Finally for $\phi_3$ the expression to be estimated is $Z_3\!=\!  PDF({last(t),0.1,0,10)})$, which corresponds to  
estimating  with what probability the value of $last(t)$ (i.e., the value of $t$ at the end of a sampled trajectory) 
falls within a discretised sub-interval of $[0,10]$. In this case we consider  $k\!=\!(10-0)/0.1\!=\! 100$ sub-intervals of $\Delta\!=\![0,10]$ 
each of width $0.1$ and with the $k$-th subinterval being $\Delta_k\!=\![0+0.1\!\cdot\! k, 0+0.1\!\cdot\!(k\!+\!1)]$ with $0\!\leq\! k\!\leq\! 99$. 
In practice, for estimating $PDF({last(t),0.1,0,10)})$, COSMOS uses $k$ internal  variables, which we may call $N^{\Delta_k}_{last(t)}$ each of which counts   
how many times the value of $last(t)$ observed at the end of a sampled trajectory $\sigma$ has been found falling into the 
$k$-th interval $\Delta_k$. The probability that $last(t)\IN\Delta_k$ 
 then simply corresponds to dividing  $N^{\Delta_k}_{last(t)}/n$, where $n$ is the total number of sampled trajectories. 
 Thus the output of estimating $Z_3\!=\!  PDF({last(t),0.1,0,10)})$ produced by COSMOS is the $k$-tuple of 
 variables $(N^{\Delta_0}_{last(t)}/n$, \ldots, $N^{\Delta_k}_{last(t)}/n$, \ldots $N^{\Delta_99}_{last(t)}/n)$. 

\subsection{COSMOS statistical model checker}
\label{sec:cosmos}
  \cosmos\footnote{\cosmos\ is an acronym of the french
  sentence ``\emph{Concept et Outils Statistiques pour le MOd\`eles
    Stochastiques}'' whose english translation would sound like:
  ``Tools and Concepts for Statistical analysis of stochastic
  models".}~\cite{BDDHP-qest11} is a  prototype software platform for
\Lang-based statistical model checking.  
It employs    \emph{confidence  interval} techniques for estimating the mean value of relevant 
performance measures expressed in terms of HASL formulae against  a given GSPN model. 
\cosmos\ has been recently integrated in the CosyVerif platform~\cite{cosyverif} 
which adds to the original command line interface (available with the first version) 
 the possibility of drawing the input elements (i.e. GSPN and LHA) through a user a graphical interface. 
Software  platforms featuring statistical model 
checking functionalities similar to \cosmos\ include:  \prism~\cite{prism}, \uppaalsmc~\cite{UPPAAL-2012}, 
and \plasma~\cite{JegourelLS12}, 
\apmc~\cite{herault2006apmc}, \ymer~\cite{younes2005ymer},
\mrmc~\cite{mrmc2} and \vesta~\cite{sen2005vesta}.  We refer the reader to~\cite{cosmos,BDDHP-qest11} for more details on  
\cosmos.

\section{Measuring   oscillations with  HASL}
\label{sec:oscillationsHASL}

\begin{figure*}[ht]
\centering
\subfigure[{\scriptsize a regular oscillation centred at 1 with maxima at 2, minima at 0, and period equal to 2}]{
\includegraphics[scale=.2]{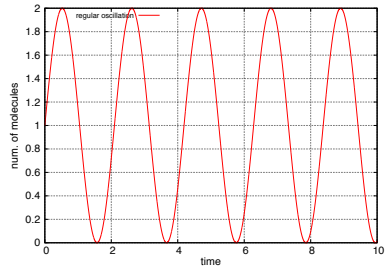}
\label{fig:regoscill}
}
\centering
\hspace{8ex}
\subfigure[{\scriptsize noisy oscillation: by considering a lower and higher thresholds we can characterise noisy-periodicity}]{
\includegraphics[scale=.2]{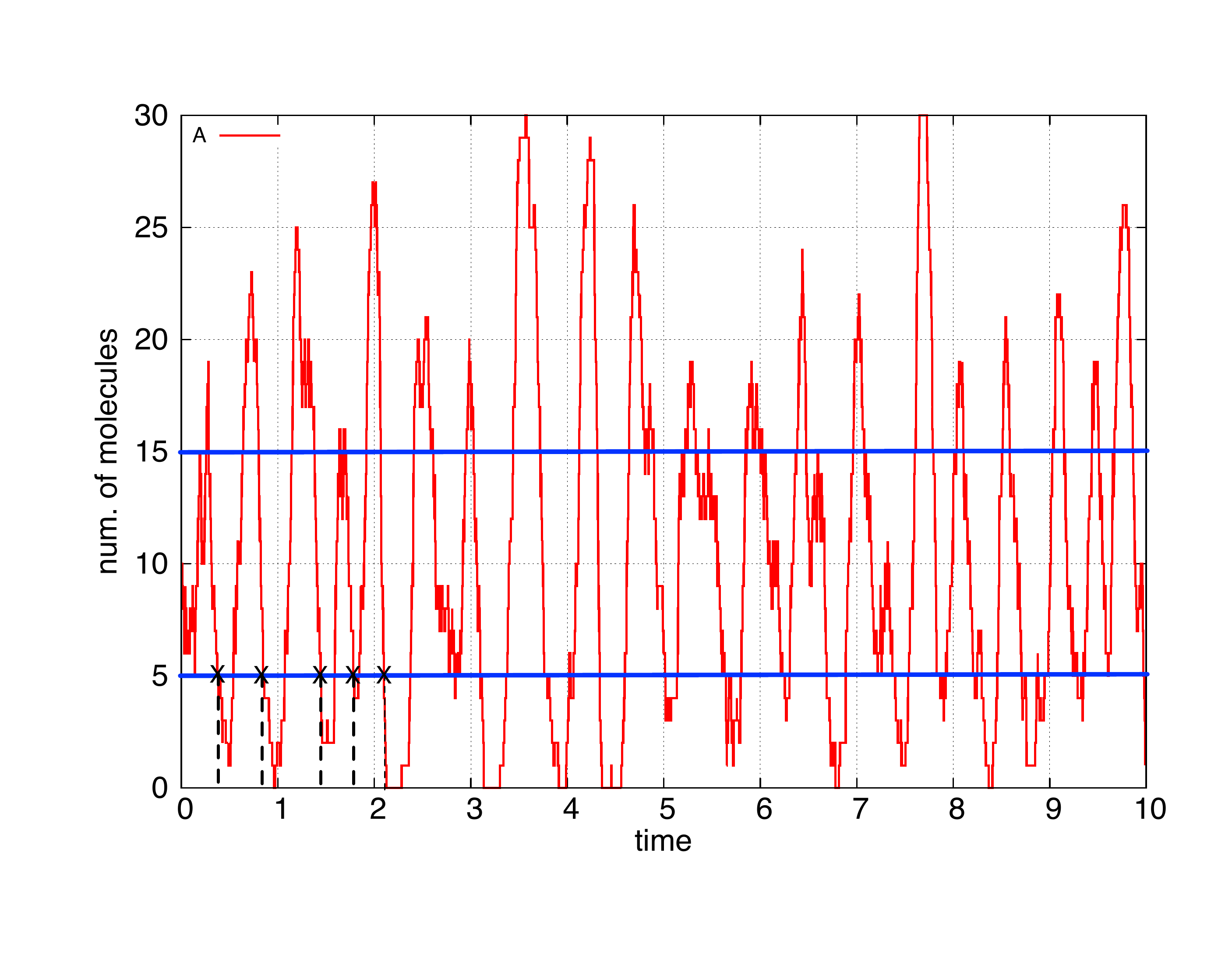}
\label{fig:noisyoscill}
}
\caption{Deterministic versus stochastic (noisy) oscillations }
\label{fig:oscillatingTraces}
\end{figure*}


Intuitively an oscillation is the periodic variation of a quantity around a given value. In mathematical terms this is associated with 
the definition of (non-constant) periodic function. i.e.   function $f:\mathbb{R}^+\to\mathbb{R}$ for which  $\exists P\in\mathbb{R}^+$ such that $\forall t\in\mathbb{R}^+$, $f(t)=f(t+P)$, where $P$ is called the period (e.g.  trace in Figure~\ref{fig:regoscill}). 
In the context of stochastic models  such a ``deterministic'' characterisation of   periodicity   is of little relevance, as 
the trajectories of a stochastic oscillator  being strictly periodic (as in $f(t)=f(t+P)$), will have (unless in degenerative cases) zero probability.  
More generally the traces of   (discrete-state) stochastic oscillators   are  characterised by a remarkable level of noise (e.g. trace in Figure~\ref{fig:noisyoscill}).  

For a stochastic model,  oscillation 
can either be either a \emph{transient} behaviour (a model which oscillates for a finite duration) or a
 \emph{limiting behaviour} (i.e. a model that oscillate sustainably for $t\to\infty$). 
 Spieler~\cite{Spieler13}, whose work tackles  CSL based analysis of sustained CTMC oscillators,  characterised sustainable oscillations as the absence of both \emph{divergence} and \emph{convergence}, 
 meaning that a (discrete-state) stochastic model that oscillates sustainably is one whose trajectories $\sigma$ cannot diverge 
 ($\lim_{t\to\infty} \sigma(t)\!<\!\infty$) nor converge  \\
 ($\nexists n\IN\mathds{N}: \lim_{t\to\infty} \sigma(t)\!=\!n$). 
 
  means that 
that is:  a model oscillates (sustainably) if and only if the probability measure of the  converging trajectories and diverging trajectories is null~\cite{Spieler13}. 
In order to study the dynamics  of stochastic oscillators, in the following  we introduce two (orthogonal) characterisations of 
oscillatory traces. The first one (named  \emph{noisy periodicity}~\cite{Spieler13}) allows us for 
observing  the period duration of an oscillator, while the second is aimed to locating the maximal and minimal 
peaks of oscillating traces. 
We first recall the definition of trajectory of a DESP.

\subsection{Measuring the period of  oscillations}
\label{sec:periodHASL}
As we pointed out that the mathematical characterisation of periodic function is a too strict one 
for  stochastic modelling framework here we consider  an alternative characterisation of 
periodicity which is suitable for capturing the noisy nature of stochastic oscillations. 
For this we establish a partition of a DESP state-space induced by two threshold  levels $L,H\IN\mathds{N}$ with $L\!<\! H$ 
and we  say that, with respect to a specific observed species (i.e. one of the $n$ dimensions of the DESP)  
a trajectory oscillates or, equivalently  is \emph{noisy periodic},  if it traverse

\begin{definition}[noisy periodic trajectory]
\label{def:noisyperiodicity}
A trajectory $\sigma$   of an $n$-dimensional DESP ${\cal D}$ population model is said \emph{noisy periodic} with respect to the $i^{th}$ 
($1\!\leq i\!\leq\! n$)   observed species  of ${\cal D}$ 
and amplitude levels $L,H\IN\mathds{N}$, with $L\!<\! H$, if $\sigma_i$ visits the intervals    $low=(-\infty,L)$, $mid=[L,H)$ and $high=[H,\infty)$ infinitely often.  
\end{definition}

In the remainder rather than referring to the periodicity with respect to the $i^{th}$ dimension 
we refer to the periodicity with respect to the population of species $A$, where $A$ is the symbolic name of the observed species 
corresponding to one of 
the Petri-net  place in the GSPN representation of  ${\cal D}$. 
Thus, with a slight abuse of notation, we will denote $\sigma_A$ a trace which is noisy periodic 
w.r.t. species $A$. 


Given a  noisy periodic trace    we are interested in 
measuring the basic  characteristics of its oscillatory nature, such as, 
the (average) duration of the oscillation \emph{period}. For this we  first  need to establish what 
we mean by \emph{period}. 
Intuitively a period, for a trace which is noisy periodic  (in the sense of Definition~\ref{def:noisyperiodicity}),  
corresponds to the time interval  between two consecutive 
sojourns  in one  of the two  extreme  regions of the partition (e.g.,  $low$ region),  
interleaved by (at least) one  sojourn  into the opposite region (e.g.,  $high$ region). 
Figure~\ref{fig:noisyperiodictrace} illustrates an example of \emph{period realisations} over a 
noisy periodic trace: the  first two period realisations, denoted $p1$ and $p2$, 
are delimited by the $mid$-to-$low$ crossing points corresponding to the first entering of the $low$ region 
which follows a previous sojourn in the $high$ region. 
Such an intuitive description of \emph{noisy period} of a noisy periodic trace  is formalised in Definition~\ref{def:noisyperiod}. 
We first introduce the notion of crossing points sets associated to a noisy periodic trace. 

Given a  noisy periodic trace $\sigma_A$  we  denote $\tau_{j \downarrow}$ (respectively  $\tau_{j\uparrow}$), 
the instant of time when  $\sigma_A$ enters for the $j$-th time the $low$ (respectively the $high$) region. 
 $T_{\downarrow}\!=\!\cup_j \tau_{j\downarrow}$ (resp. $T_{\uparrow}\!=\!\cup_j \tau_{j\uparrow}$) 
is the set of all \emph{low-crossing points} (reps. \emph{high-crossing points}). 
Observe that $T_{\downarrow}$ and $T_{\uparrow}$ 
reciprocally induce a partition on each other. Specifically  $T_{\downarrow}\!=\!\cup_k T_{k\downarrow}$ 
where $T_{k\downarrow}$ is the  subset of  $T_{\downarrow}$ containing 
the $k$-th sequence of contiguous  \emph{low-crossing points} not interleaved by 
any \emph{high-crossing point}. Formally\\ $T_{k\downarrow}\!=\!\{\tau_{i\downarrow}, \ldots, \tau_{(i+h)\downarrow} | \exists k',  \tau_{(i-1)\downarrow}\!<\! \tau_{k'\uparrow}\!<\! \tau_{i\downarrow}, \\  \tau_{(i\!+\!h)\downarrow}\!<\!\tau_{(k'\!+\! 1)\uparrow} \}$. 
Similarly $T_{\uparrow}$ is partitioned  $T_{\uparrow}\!=\!\cup_k T_{k\uparrow}$ where 
$T_{k\uparrow}$ is the  subset of  $T_{\uparrow}$ containing 
the $k$-th sequence of contiguous  \emph{high-crossing points} not interleaved by 
any \emph{low-crossing point}. 
For example, with respect to  trace $\sigma_A$ depicted in Figure~\ref{fig:noisyperiodictrace} 
we have that $T_{\downarrow}\!=\! T_{1\downarrow}\!\cup\! T_{2\downarrow} \!\cup\! T_{3\downarrow}\ldots $
with $T_{1\downarrow}\!=\! \{\tau_{1\downarrow}, \tau_{2\downarrow}\}$, 
$T_{2\downarrow}\!=\! \{\tau_{3\downarrow}\}$, $T_{3\downarrow}\!=\! \{\tau_{4\downarrow}\}$, 
while $T_{\uparrow}\!=\! T_{1\uparrow}\!\cup\! T_{2\uparrow} \!\cup\! T_{3\uparrow}\ldots $ 
with $T_{1\uparrow}\!=\! \{\tau_{1\uparrow}\}$, $T_{2\uparrow}\!=\! \{\tau_{2\uparrow}\}$, 
$T_{3\uparrow}\!=\! \{\tau_{3\uparrow}\}$. 
\begin{figure}[ht]
\centering
\includegraphics[scale=.65]{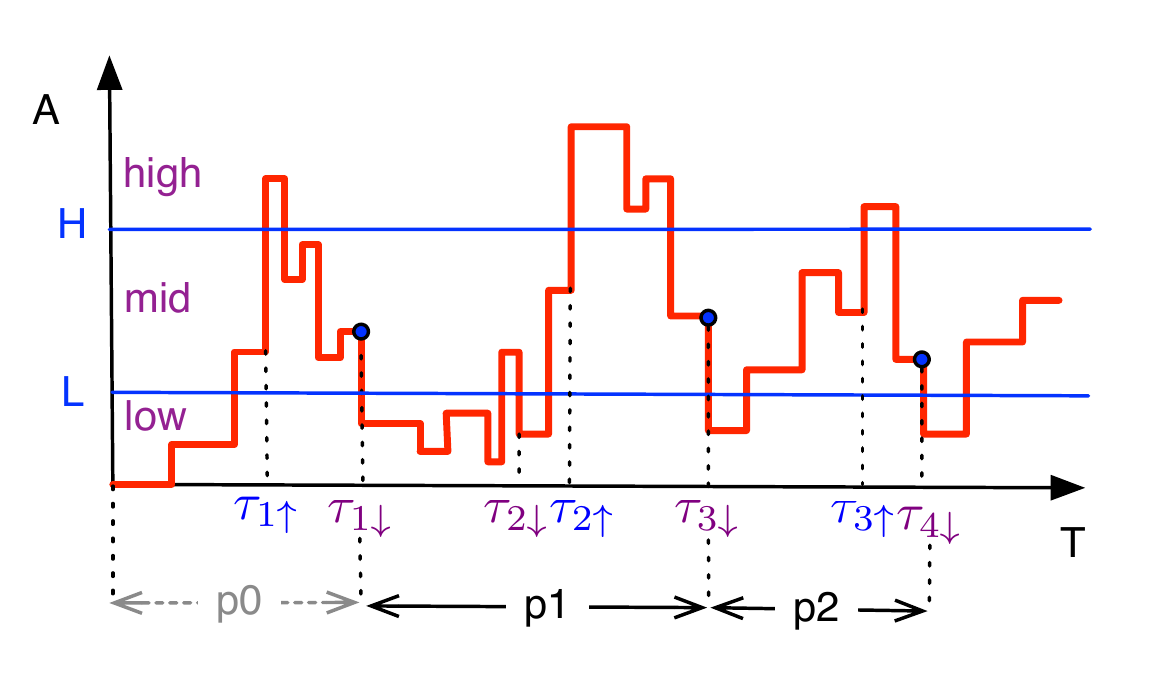}
\caption{Example of  trace $\sigma_A$  which is  {\it noisy periodic} w.r.t to species $A$ and  a given ($L\!<\!H$ induced) partition of a DESP  state space. 
} 
\label{fig:noisyperiodictrace}
\end{figure}
Observe that   a noisy periodic trace can be seen  as 
a collection of realisations  of certain random variables. 
For example the instants of time  
   $\tau_{j \downarrow}$,  $\tau_{j\uparrow}$ are realisations of 
the  random variables (which we could denote $x_{j \downarrow}$, respectively $x_{j \uparrow}$) 
corresponding to the timing of    entering the $low$, respectively  $high$,  regions. 
Similarly the duration of the $k$-th period contained in a trace can be seen  
as the realisation of a random variable\footnote{Note that the duration of the $k$-th period of a trace is,  in turn, dependent on the  the random variables $x_{j \downarrow}$, $x_{j \uparrow}$ corresponding to the entering of the $low$, $high$ regions.}. 
We formalise the notion of noisy period realisation in the next definition. 

\begin{definition}[$k^{th}$noisy period realisation]
\label{def:noisyperiod}
For $\sigma_A$  a noisy periodic trajectory with crossing point times $T_{\downarrow}\!=\!\cup_{k\geq 1} T_{k\downarrow}$ , 
respectively $T_{\uparrow}\!=\!\cup_{k\geq 1} T_{k\uparrow}$, 
 the   realisation of the $k^{th}$ noisy period, denoted $t_{p_k}$, is defined as  $t_{p_k}\!=\! min(T_{(k\!+\!1)\downarrow})- min(T_{k\downarrow})$.
\end{definition}

\noindent  
Observe that a  noisy periodic trace (as of Definition~\ref{def:noisyperiodicity}) contains infinitely many realisations of (noisy) periods. 
In the remainder we will  refer to the $N$-prefix of a noisy periodic trace $\sigma_A$, meaning the prefix of $\sigma_A$ 
that consists of the first $N$ noisy period realisations. 

As an example of period realisations, let us consider  the noisy periodic trace  in Figure~\ref{fig:noisyperiodictrace}  
whose first two  period realisations  are $t_{p_1}\!=\! \tau_{3\downarrow}-\tau_{1\downarrow}$
and $t_{p_2}\!=\! \tau_{4\downarrow}-\tau_{3\downarrow}$. 
Notice that the time interval denoted as $p0$ in Figure~\ref{fig:noisyperiodictrace} 
does not represent a complete period realisation as 
 there's no guarantee that   $T=0$ corresponds with the actual entering into the $low$ region.  
Definition~\ref{def:noisyperiod} correctly does not account for the first \emph{spurious} period $p0$.

Having introduced the notion of  noisy period realisation 
 we now look at the problem of 
estimating two characteristic measures related to it, namely, the \emph{period average}  and the \emph{period fluctuation}. 
By \emph{period average} we  simply  
mean the average value of the  period realisations sampled along a trace. 
On the other hand by \emph{period fluctuation} we mean a measure of the variability of 
the period realisations along a trace, that is,  a measure of how much  periods observed along 
a trace  vary one another. 
Observe that, from the point of view of analysis,  \emph{period fluctuation}
allows us to analyse  the regularity of the observed oscillator. 
In this respect a ``regular'' oscillator is one whose traces consists of little variable periods (i.e.,  small fluctuation), 
as opposed to an ``irregular'' one whose traces exhibits variable periods (i.e.,  large fluctuation). 
We demonstrate the analysis of oscillation regularity through fluctuation assessment in Section~\ref{sec:circlock}).

\begin{definition} [period average]
\label{def:averagenoisyperiod}
For $\sigma_A$  a noisy periodic trajectory 
the  period average of the first $n\!\in\! \mathbb{N}$ period realisations, denoted $\overline{t}_{p}(n)$, is defined as $\overline{t}_{p}(n)\!=\!\frac{1}{n}\sum_{k=1}^n t_{p_k}$, 
where $t_{p_k}$ is the $k$-th period realisation. 
\end{definition}

Observe that, in the long run, the   average value of the noisy-period of oscillations corresponds to the limit $\overline{t}_{p}=\lim_{n \to \infty}\overline{t}_{p}(n)$. 

\begin{definition}[period fluctuation]
\label{def:flcutnoisyperiod}
For $\sigma_A$  a noisy periodic trajectory 
the  period fluctuation of the first $n\!\in\! \mathbb{N}$ period realisations, denoted $s^2_{t_p}(n)$, is defined as $s^2_{t_p}(n)\!=\!\frac{1}{n}\sum_{k=1}^n (t_{p_k}-\overline{t}_{p}(n))^2$, 
where $t_{p_k}$ is the $k$-th period realisation and $\overline{t}_{p}(n)$ is the period average for the first $n$ period realisations. 
\end{definition}

Note that the period fluctuation is in essence defined as the variance of the period realisations along a trace. 
In the remainder we show how, through  automaton ${\cal A}_{per}$, we can estimate the 
period fluctuation \emph{on-the-fly}, that is, as the noisy periodic traces are  generated and 
scanned by ${\cal A}_{per}$. For this we employ an  
adaptation  of the so-called \emph{online algorithm}~\cite{Knuth:1997:ACP:270146} for computing the variance out of a sample of observations.   



\begin{figure*}[ht]
\centering
\fbox{
\includegraphics[scale=.35]{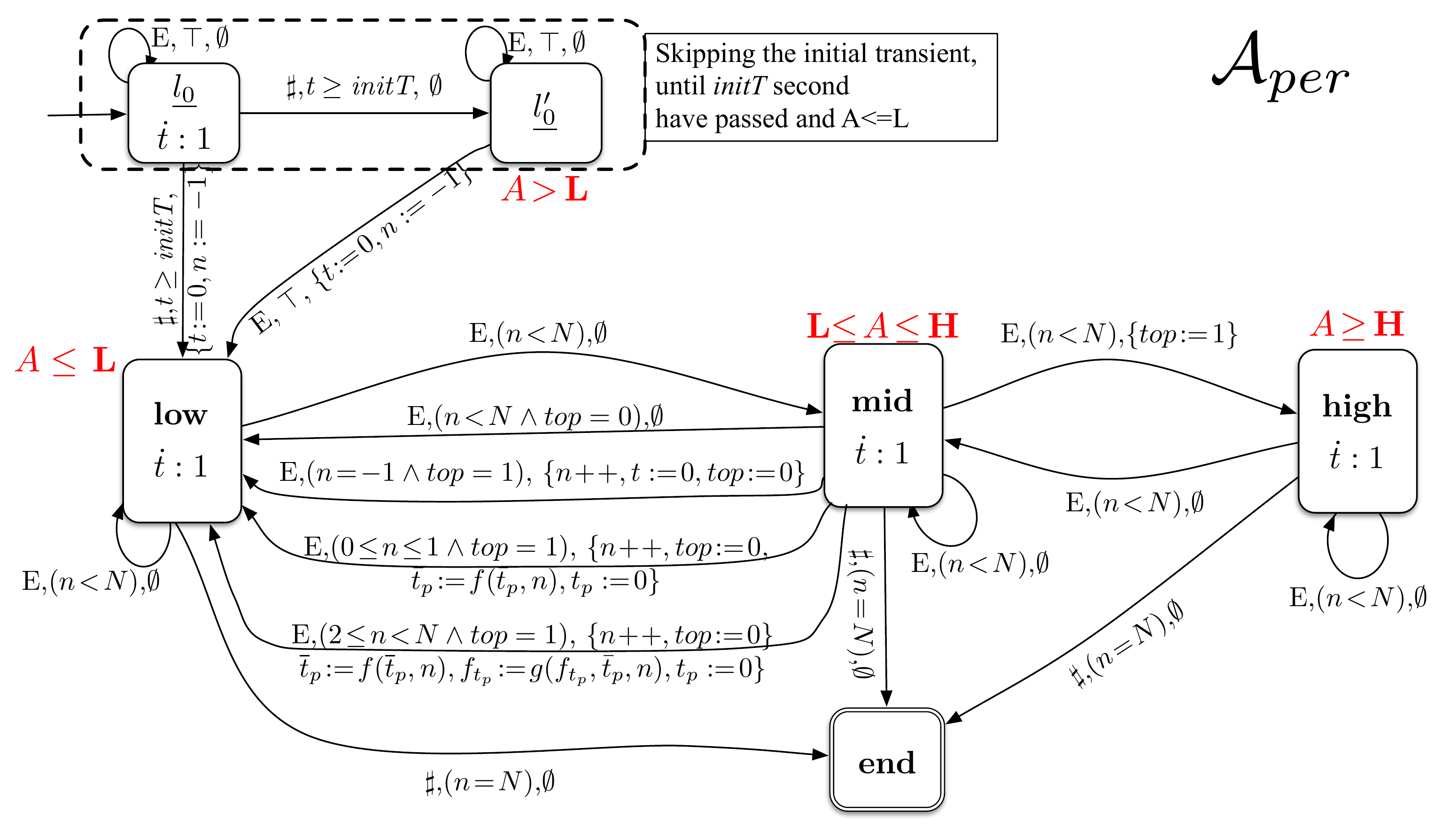}}
\caption{ ${\cal A}_{per}$: an LHA for selecting noisy periodic traces (with respect to an observed species $A$) related to  partition $low=(-\infty,L]$, $mid=(L,H)$ and $high=[H,+\infty)$. }
\label{fig:lhaoscill}
\end{figure*}

In the following we  introduce 
an LHA automaton, called ${\cal A}_{per}$, which is targeted to estimating  
both the average and the fluctuation of the first $N$    the 
period realisations  occurring along the simulated noisy periodic traces.

\paragraph{The automaton ${\cal A}_{per}$.} The LHA  ${\cal A}_{per}$ depicted in Figure~\ref{fig:lhaoscill} is   designed for detecting noisy periods of an observed species (here denoted  $A$). 
It consists of an initial \emph{transient filter} (locations $l_0$, $l'_0$) plus three main locations {\bf low}, {\bf mid} and {\bf high} 
(corresponding to the    partition of $A$'s domain induced by thresholds $L<H$). 
The intuition behind the structure of the ${\cal A}_{per}$ automaton is as follows: 
the transient filter is used to simply let the simulated trajectory unfold for a given duration (which is useful for 
eliminating the effect of the initial transient from long measures, see below). 
On the other hand the actual analysis of the periodicity is performed by  
looping within the {\bf low}, {\bf mid} and {\bf high}  locations. In particular 
each of these three locations corresponds to one region of the partition induced by the considered 
$L\!<\! H$ thresholds: location {\bf low} corresponds to region $low=(-\infty,L]$, location {\bf mid} to region $mid=(L,H)$ 
and location {\bf high} to region  $high=[H,+\infty)$. Thus while a trace of the considered DESP is simulated 
the ${\cal A}_{per}$ automaton oscillates in between locations {\bf low}  and {\bf high}, passing through {\bf mid}, following 
the profile of the observed species $A$. The completion of a loop from {\bf low} to  {\bf high} and back to {\bf low}  corresponds 
to detection of a period realisation (as of Definition~\ref{def:noisyperiod}) on occurrence of which a number 
of relevant information is stored in the data variables of ${\cal A}_{per}$. The analysis of the 
simulated trajectory ends by entering location {\bf end} as soon as the $N$-th period has been detected. 
Below we provide a more detailed description of the functioning of ${\cal A}_{per}$.

 \begin{table*}
 \begin{center}
     \begin{tabular}{ | c | c | c | p{4.5cm} |}
    \hline
    \multicolumn{4}{|c|}{Data variables} \\
    \hline
    name & domain  & update definition & description  \\ \hline
    
        $t$ & $\mathbb{R}_{\geq 0}$ & \emph{reset} & time elapsed since beginning measure (first non-spurious period) \\ \hline
      $n$ & $\mathbb{N}$ & \emph{increment} & counter of  detected periods \\ \hline
      $top$ & bool & \emph{complement} & boolean flag indicating whether the high part of the partition has been entered \\ \hline
    
     ${t}_p$ & $\mathbb{R}_{\geq 0}$ &  \emph{reset}  & duration of last detected period  \\ \hline
   
    $\bar{t}_p$ & $\mathbb{R}_{\geq 0}$ & $  f(\bar{t}_p,t_p,n)=\frac{\bar{t}_{p_n}\cdot n+t_p}{n+1}$ & mean value of $t_p$ \\
    
            & & & \\ \hline
    
        & & & \\
    $s^2_{t_p}$ & $\mathbb{R}_{\geq 0}$ & $g(s^2_{t_{p}},\bar{t}_p,t_p,n)  =\frac{[(n-1) s^2_{t_{p}}+ (t_p-\bar{t}_p)(t_p\! -\! f(\bar{t}_p,t_p,n+1))]}{n}$ & fluctuation  of $t_p$      \\         
    & & & \\ \hline
    \end{tabular}
\end{center}
    \caption{ The data variables  of automata ${\cal A}_{per}$ of Figure~\ref{fig:lhaoscill} for measures of noisy-periodicity}
    \label{tab:lhavars2}
    \end{table*}

The synchronization  starts  in $l_0$ where the automaton loops through synchronous edge $l_0\xrightarrow{E,\top,\emptyset} l_0$,  simply 
observing the occurrences of  any event of $E$ (i.e.,  the event set of 
synchronised GSPN-DESP model), hence   letting the simulated trajectory unfold  for a fixed duration given by parameter 
$initT$: when $t\!\geq\! initT$  ${\cal A}_{per}$ moves, through autonomous edges, to either $l'_0$, if by  $t\!=\! initT$ the simulated trajectory is not in a state 
of the $low=(-\infty,L]$ region (i.e., if the invariant  
 $A\! >\! L$ of $l'_0$ is satisfied), or to location {\bf low} if the current state of the trajectory belongs to the $low=(-\infty,L]$ region 
 (i.e., if the invariant $A\!\leq\! L$ of {\bf low} is fulfilled). 
If  $l'_0$ is entered then the simulated trace is let further unfolding  (synchronised self-loop $l'_0\xrightarrow{E,\top,\emptyset} l'_0$) 
until a state within $low=(-\infty,L]$ is reached, in which case the invariant of location {\bf low} is fulfilled hence the autonomous edge  $l'_0\xrightarrow{\sharp,\top,\ldots}${\bf low} is traversed. Observe that on entering of {\bf low} the global timer variable $t$ is reset 
and the period counter $n$ is initialised to $-1$ (this is so to avoid the first spurious  period, denoted $p0$ in Figure~\ref{fig:noisyperiodictrace}, 
to be considered amongst the detected ones). 
 Once  in location {\bf low} the actual detection of the period realisations begins\footnote{
Although the LHA in Figure~\ref{fig:lhaoscill} is designed so that  periods detection starts from $low$ it can be easily adapted so 
that the identification starts from any location.} and the automaton gets looping 
 between the {\bf low}, {\bf mid} and {\bf high}  locations for as long as $N$ periods have been detected. 
From {\bf low} the automaton follows the profile exhibited by the observed population   $A$, 
thus moving to  {\bf mid} (and possibly back) as soon as 
the population of $A$ grows and a state of the $mid=(L,H)$ region is entered (i.e., corresponding to the $L\!< A\!<\!H$ invariant of {\bf mid} location becoming satisfied), 
and then to {\bf high} (and possibly back) as soon as the population of $A$ enters the $high=[H,+\infty)$ region (corresponding to the $A\!\geq\!H$ invariant of {\bf high} location). On entering the {\bf high} location the boolean variable $top$ is set to true (i.e., $top\!=\! 1$). This allows then for  distinguishing 
between the   {\bf mid}-to-{\bf low} transitions of  kind {\bf mid}$\xrightarrow{E,(\ldots \land top=1),\ldots}${\bf low}, which correspond to an actual closure of a period realisation (i.e., those   $\tau_{j\downarrow}$  
preceded by a sojourn in the  $high=[H,+\infty)$ region), from those of  kind   
{\bf mid}$\xrightarrow{E,(\ldots \land top=0),\ldots}${\bf low} 
which correspond to a return to {\bf low} without having previously sojourned in {\bf high}. 
Observe that from  {\bf mid} location there are four possible (mutually exclusive) ways of entering 
the {\bf low} location. If the sojourn in {\bf mid} has not been preceded by a sojourn in {\bf high} 
edge {\bf mid}$\xrightarrow{E,(n\!<\! N \land top=0),\ldots}${\bf low} is enabled. 
On the other hand if the sojourn in {\bf mid} has  been preceded by a sojourn in {\bf high} but 
{\bf low} is going to be re-entered for the first time (i.e., $n=-1$) then the timer $t$ is reset (representing the 
start time of actual period detection) and the counter of detected periods $n$ is 
set to zero (again representing the actual beginning of counting of period detection). 
On the other hand if  the sojourn in {\bf mid} has  been preceded by a sojourn in {\bf high} 
and the period to be detected is the first one (i.e., $0\!\leq n\!\leq\! 1\land top\!=\! 1$) then 
we increment the counter $n$ of detected period, we reset the flag $top$ and 
update the value of the average duration of detected period $\hat{t}_p$ while we do not update 
the variable $s^2_{t_p}$ as in order as 
in order to update the value of the fluctuation of the detected period duration  we need that at least two periods have been detected. 
Finally if the period to be detected is the $n$-th with $n\!\geq\! 2$ (i.e., corresponding to guard 
$2\!\leq n\!\leq\! N\land top\!=\! 1$) we do the same update operations of the previous case but also update $s^2_{t_p}$.

The automata uses variable $n$ to count the number of  noisy periods detected along a trajectory, and stops as soon as  the $N^{th}$ period is detected (i.e. event bounded measure). 
The boolean variable $top$, which is set to $true$ on entering of the $high$ location,  allows for detecting  the completion of a period  (i.e. crossing from 
$mid$ to $low$ when $top$ is $true$). Two clock variables, $t$ and $t_p$, maintains respectively the total simulation time 
 as of the beginning of the first detected period ($t$) and the duration of the last detected period ($t_p$). Finally 
variable $\overline{t_p}$ maintain the average duration of all (so far) detected periods while 
$s^2_{t_p}$ stores the fluctuation (or variability) of duration (i.e. how far the duration of each detected period is distant from 
its average value computed along a trajectory)  of all (so far) detected periods. 

\begin{theorem}
\label{theo:Aper}
If a trace $\sigma_A$ is noisy periodic w.r.t. amplitude levels $L,H\IN\mathds{N}$ then it 
is  accepted by automaton ${\cal A}_{per}$ with parameters $L$, $H$, $initT\IN\mathds{R^+}$ and $N\IN\mathds{N}$
\end{theorem}
\begin{proof}

By hypothesis $\sigma_A$ (the  projection, w.r.t  an observed species $A$,  of a trace $\sigma$ an $n$-dimensional DESP ${\cal D}$) 
  is noisy periodic w.r.t. the partition of species $A$'s domain into regions $low=(-\infty,L)$, $mid=[L,H)$ and $high=[H,\infty)$. 
The initial state of the synchronised process ${\cal D}\times{\cal A}_{per}$  
will be $(\sigma_A[0],l_0,Val_0)$ with $Val_0$ being the initial valuation 
with $Val_0(x)\!=\!0$ for all variables $\forall x\IN X$ of $A_{per}$. 

To demonstrate that $\sigma_A$ is accepted by ${\cal A}_{per}$ we need to show that 
starting from the the initials state $(\sigma_A[0],l_0,Val_0)$  a final state of ${\cal D}\times{\cal A}_{per}$, i.e., a state 
such that the current  location is the accepting location {\bf end} of ${\cal A}_{per}$,   is reached. 
For this we proceed by induction w.r.t. 
the number of subsequent sojourns in the $low$ and $high$ regions. 
In the remainder we use the  notation $(s,l,Val)\xrightarrow{*} (s',l', Val')$ to indicate 
that state  $(s',l',Val')$ of process  ${\cal D}\times{\cal A}_{per}$ is reachable 
from  $(s,l,Val)$. 
We split the demonstration in parts corresponding to the traversal of ${\cal A}_{per}$ locations  
resulting from synchronisation with trace $\sigma_A$:

\begin{description}
\item[[{\bf init}]] Let $l_0\IN\mathds{N}$ be the index of the first state 
of $\sigma_A$ that belongs to $low$ and that follows $\sigma_A@initT$ (where $initT$ is the parameter of ${\cal A}_{per}$), that is: 
$l_0\!=\!min\{i\IN\mathds{N}\mid i> i_{initT}\land \sigma[i]\IN low\}$, with $i_{initT}$ being the index of 
the state $\sigma_A$ is in at time $initT$ (observe that since $\sigma_A$ is assumed 
noisy periodic then $min\{i\IN\mathds{N}\mid i> i_{initT}\land \sigma[i]\IN low\}$ 
is guaranteed to exist).  
Thus because of the structure of ${\cal A}_{per}$,  
$(\sigma_A[0],l_0,Val_0)\xrightarrow{*} (\sigma[l_0],\text{\bf low}, Val_{l_0})$, 
with $Val_{l_0}(x)\!=\!0$, for $ x\!\neq\! n $ and $Val_{l_0}(n)\!=\!-1$. \\

\item [[{\bf low}$\to${\bf mid}$\to${\bf high}]\hskip -1ex] Since $\sigma_A$ is noisy periodic then  $\exists m_1,h_1\IN\mathds{N}:h_1\!>\!m_1\!>\!l_0$ such that $\sigma_A[ m_1]\IN mid$, $\sigma_A[ h_1]\IN high$ 
hence, because of the structure of ${\cal A}_{per}$,  it follows  that $(\sigma[l_0],\text{\bf low}, Val_{l_0})\xrightarrow{*} (\sigma[m_1],\text{\bf mid}, Val_{m_1}) 
\xrightarrow{*} (\sigma[h_1],\text{\bf high}, Val_{h_1})$ 
with $Val_{h_1}(top)=1$, because of the update $\{top:=1\}$ of the arc  leading to   location {\bf high}. \\

\item [[{\bf high}$\to${\bf mid}$\to${\bf low}]\hskip -1ex] similarly since $\sigma_A$ is noisy periodic then  $\exists l_1,m_{1b}\IN\mathds{N}:l_1\!>\!m_{1b}\!>\!h_1$ such that $\sigma_A[ m_{1b}]\IN mid$, $\sigma_A[ l_1]\IN low$. 
hence $(\sigma[h_1],\text{\bf high}, Val_{h_1})\xrightarrow{*} (\sigma[m_{1b}],\text{\bf mid}, Val_{m_{1b}}) 
\xrightarrow{*} (\sigma[l_1],\text{\bf low}, Val_{l_1})$ 
with $Val_{m_{1b}}(top)=1$, $Val_{m_{1b}}(n)=-1$ hence  $Val_{l_1}(top)=0$, $Val_{l_1}(n)=0$, $Val_{l_1}(t)=0$, since, 
because of   $Val_{m_{1b}}$, location {\bf low} is entered through edge \\ 
{\bf mid}$\xrightarrow{E,(n\!=\! -1 \land top=1),\{n+\!+,top:=\!0,t:=\!0\}}${\bf low} \\

\item [{\bf induction}] since $\sigma_A$ is noisy periodic then the $high$, $low$ regions are entered infinitely often, 
hence the [{\bf low}$\to${\bf mid}$\to${\bf high}] and [{\bf high}$\to${\bf mid}$\to${\bf low}] steps of the proof hold for each successive  
iteration. This means that if $l_i$ is the index corresponding to the $i$-th that  $\sigma_A$ enters the $low$ region 
after having sojourned in the $high$ region then because of the periodicity of $\sigma_A$ 
$\exists m_{i+1},h_{i+1},m{(i+1)b}\IN\mathds{N}:l_{i+1}\!>\!m_{(i+1)b}\!>\!h_1\!>\!m_1\!>\!l_i$ such that 
$\sigma_A[ m_{ib}],\sigma_A[ m_{(i+1)b}]\IN mid$, $\sigma_A[ l_i], \sigma_A[ l_{i+1}]\IN low$, $\sigma_A[ h_{(i+1)}]\IN high$, 
hence $(\sigma[l_i],\text{\bf low}, Val_{l_i})\xrightarrow{*} (\sigma[l_{i+1}],\text{\bf low}, Val_{l_{i+1}})$, 
with $Val_{l_{i+1}}(n) \!=\! Val_{l_{i}}(n) +1$.  

\item [{\bf termination}, [{\bf low}$\to${\bf end}]\hskip -1ex] 
By induction we have seen that $\forall i\IN\mathds{N}$, 
 $(\sigma[l_i],\text{\bf low}, Val_{l_i})\xrightarrow{*} (\sigma[l_{i+1}],\text{\bf low}, Val_{l_{i+1}})$. Thus on the $(N-1)$-th iteration 
 $(\sigma[l_{N-1}],\text{\bf low}, Val_{l_{N-1}})\xrightarrow{*} (\sigma[l_{N}],\text{\bf low}, Val_{l_{N}})$ with  $Val_{l_N}(n)=N$ which enables  
{\bf low}$\xrightarrow{\sharp,(n\!=\! N ),\emptyset}${\bf end}, hence 
 $(\sigma[l_{N}],\text{\bf low}, Val_{l_{N}})\xrightarrow{*} (\sigma[l_{N}],\text{\bf end}, Val_{l_{N}})$ and 
 $\sigma_A$ is accepted.

\end{description}

\flushright{$\blacksquare$}

\end{proof}

\paragraph{HASL expressions associated to ${\cal A}_{per}$.}
We  define different HASL expressions to be associated to to automaton ${\cal A}_{per}$. 

\begin{itemize}
\item $Z_1\equiv E[last(\bar{t}_p)]$: corresponding to  the mean value of the period duration  for the first $N$  detected periods. 
\item $Z_{2}\equiv PDF(\bar{t}_p,s, l, h)$: corresponding to the
  PDF of the average period duration over the first $N$  detected periods, where $[l,h]$ represents  the  considered support 
of the estimated PDF, and $[l,h]$ is  discretized into uniform subintervals of width $s$  
\item $Z_3\equiv E[last(s^2_{t_p})]$: corresponding to the   fluctuation
  of the period duration. 
\end{itemize}




Expression $Z_1$ represents the expected value assumed by variable $\bar{t}_p$, that is, the average 
duration of the first $N$  periods  detected along a trace, at the end of 
accepted trajectory (i.e., a trajectory that contains $N$ periods). 
Similarly expression $Z_{2}$  evaluates the PDF of  
the average duration of the first $N$  periods by assuming the interval $[l,h]$ as the  support of the PDF 
 and considering that $[l,h]$ is discretised in $(h-l)/s$ uniform subintervals of width $s$. On the other hand $Z_3$
is concerned with  assessing the expected value that variable $s^2_{t_p}$ has 
at the end of a trace consisting of $N$ noisy periods. By definition 
(see Table~\ref{tab:lhavars2})  $s^2_{t_p}$  corresponds to the 
 \emph{fluctuation} of the  duration of the detected periods, (i.e.,  how much the $N$ periods detected along a trace differ from their average  duration). 
Observe that  the measured period fluctuation  (i.e. $Z_3$) provides us
with a useful measure of the \emph{irregularity}, from the point of 
view of the period duration,  of the observed oscillation. 

\subsection{Measuring the peaks of oscillations}

In the previous section we have seen how 
a characterisation of periodicity for stochastic oscillation can be 
obtained by considering a given partition, induced by two thresholds $L,H$,  of the domain of the observed population. 
The drawback of such a characterisation is that, the detected periods depend  on the chosen  $L,H$ thresholds, and these  
have to be chosen by the modeller manually, i.e., normally by looking at the shape of a sampled trajectory 
and then choosing where to ``reasonably''  set the $L$ and $H$ values before executing the 
measurements with automaton ${\cal A}_{per}$. 
To improve things here we propose a different approach 
which is aimed at identifying where the peaks (i.e., the local maxima/minima) of 
 oscillatory  traces are located. 

  Since traces of a DESP consist of discrete  increments/decrements of at least one unit,  
it is up to the observer to establish what should be accounted for as a local maximum (minimum) during such detection process. 
Intuitively  a local max/min of a trace $\sigma_A$ (the projection of $\sigma$ w.r.t.  the observed species  $A$) 
 is a state $\sigma_A[i]$ ($i\!\in\!\mathds{N}$) that corresponds to a change of trend in the population of $A$. This is 
 formally captured by the following definition. 
\begin{definition}[local maximum/mininimum of a trace]
\label{def:maxmin}
For $\sigma_A$ the $A$ projection of a trace $\sigma$    of an $n$-dimensional DESP ${\cal D}$ population model,  
 state $\sigma_A[i]$ is a maximum,  if  $\sigma_A[i-1]\!<\!\sigma_A[i]\!>\! \sigma_A[i+1] $, 
or a minimum,  if $\sigma_A[i-1]\!>\!\sigma_A[i]\!<\! \sigma_A[i+1] $. 
\end{definition}

\begin{figure}[ht]
\centering
\includegraphics[scale=.65]{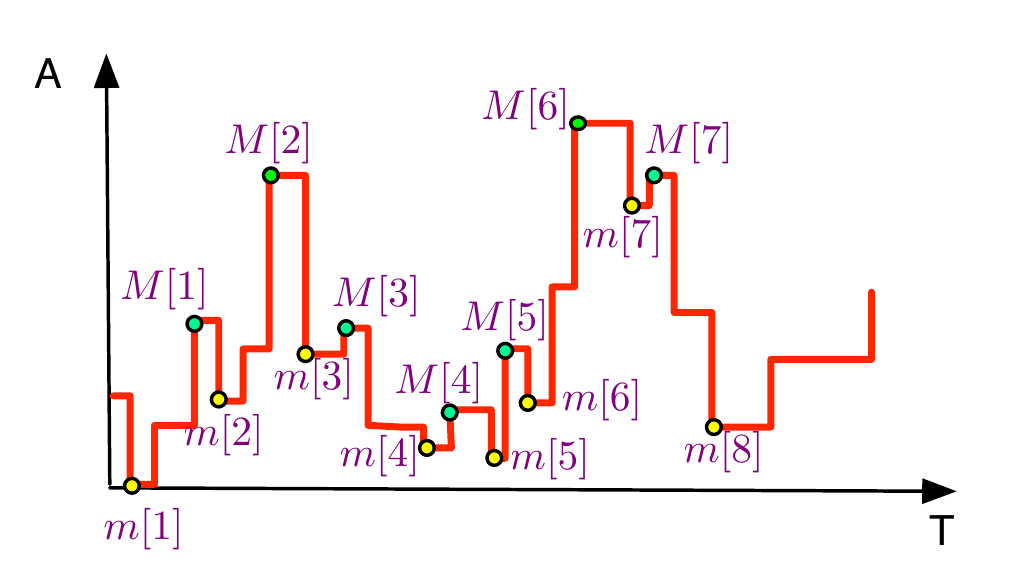}
\caption{Example of  local maxima/minima of an alternating trajectory $\sigma_A$.} 
\label{fig:simpleminmaxtrace}
\end{figure}

In the remainder we refer to a trace that consists of an infinite sequence of local maxima interleaved with 
an infinite sequence of local minima as an \emph{alternating trace} (Definition~\ref{def:alternatingtrace}). 
\begin{definition}[alternating trajectory]
\label{def:alternatingtrace}
A trajectory $\sigma$   of an $n$-dimensional DESP ${\cal D}$ population model is said \emph{alternating} with respect to the $i^{th}$ 
($1\!\leq i\!\leq\! n$)   observed species  of ${\cal D}$, if $\sigma_i$  contains infinitely many 
local minima (or equivalently local maxima).  
\end{definition}

For $\sigma_A$ an alternating trace we denote $\sigma^M_A\!=\! M[1],M[2],\ldots$, respectively $\sigma^m_A\!=\! m[1],m[2],\ldots$, 
the projection of  $\sigma_A$ consisting of the local maxima, respectively minima, of $\sigma_A$. 
Figure~\ref{fig:simpleminmaxtrace} shows the local maxima and minima for an example of alternating trace $\sigma_A$. 
In the following we point out two simple properties relating the definition of noisy periodic and alternating trace.

\begin{proposition}
\label{prop:noisyper}
If $\sigma_A$ is  a \emph{noisy periodic}  trace (as of  Definition~\ref{def:noisyperiodicity}) 
then it is also  \emph{alternating}. 
Observe however that the opposite is not necessarily true, in fact an alternating  trace  
may diverge, in which case it does not oscillate. 
\end{proposition} 
Proposition~\ref{prop:noisyper} is trivially true as by definition a noisy period  trace 
visit infinitely often the $low$ and $high$ region of the state space, thus necessarily 
it contains an infinite sequence local maxima interleaved with  local minima. 

\begin{corollary}
\label{cor:altern}
If $\sigma_A$ is an \emph{alternating} trace (as of  Definition~\ref{def:noisyperiodicity}) 
then it is not necessarily  \emph{noisy periodic}. 
\end{corollary} 
Corollary~\ref{cor:altern} simply points out that, by definition, an alternating trajectory may be diverging 
(for example if it consist of increasing steps which are always larger than the decreasing ones), in which 
case clearly  it is not noisy periodic.

In the remainder we introduce a HASL based procedure for detecting the local maxima and local minima of 
alternating traces. 
However rather than considering  detection of ``simple'' local maxima/minima as of Defintion~\ref{def:maxmin}, 
we refer to detection of a generalised notion of local maxima/minima of a trace, that is, 
maxima and minima which are distanced, at least ,by a certain value $\delta$. We formalise this notion 
in the next definition. 

\begin{definition}[$\delta$-separated local maxima]
\label{def:maxmindelta}
Let $\delta\!\in\!\mathds{R}^+$, and $\sigma_A$ the $A$ projection of a trace $\sigma$    of an $n$-dimensional DESP ${\cal D}$ population model,  a state $\sigma_A[i]$ is the $j$-th, $j\!\in\!\mathds{N}_{> 0}$,   $\delta$-separated local maximum (minimum), denoted $M_{\!\delta}[j]$ ($m_{\!\delta}[i]$), 
if it is the largest  local maximum (minimum) 
whose distance from the preceding local minimum $m_{\!\delta}[j-1]$ (maximum $M_{\!\delta}[j-1]$)  is at least $\delta$. 
\end{definition}
For $\sigma_A$ an alternating trace we denote \\
$\sigma^{M_\delta}_A\!=\! M_\delta[1],M_\delta[2],\ldots$, respectively \\
$\sigma^{m_\delta}_A\!=\! m_\delta[1],m_\delta[2],\ldots$, 
the projection of  $\sigma_A$ consisting of the $\delta$-separated local maxima, respectively minima, of $\sigma_A$. 

\begin{figure}[ht]
\centering
\includegraphics[scale=.65]{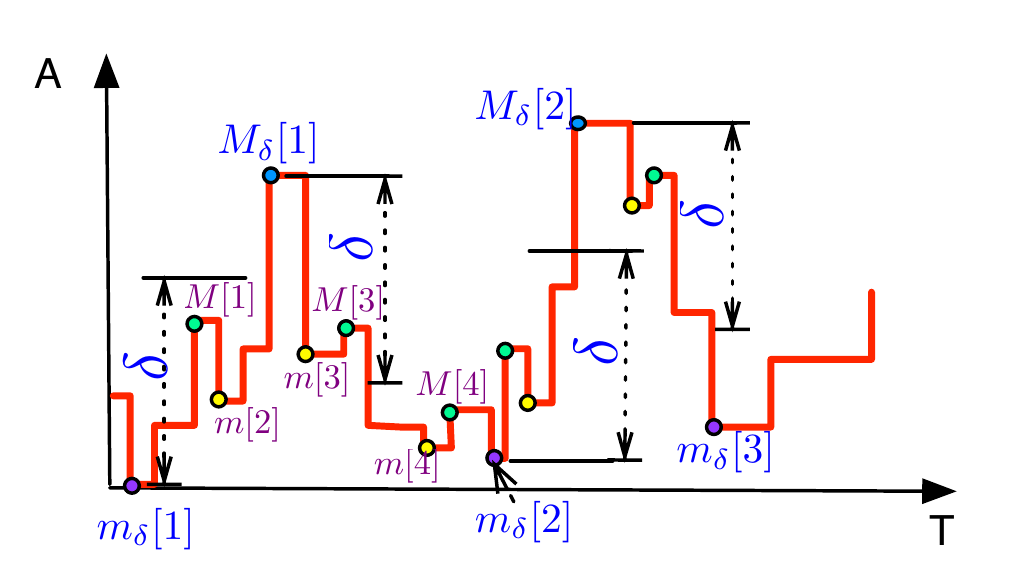}
\caption{Example of  $\delta$-separated local maxima ($M_\delta[i]$) and minima ($m_\delta[i]$) of an alternating trajectory $\sigma_A$}
\label{fig:minmaxtrace}
\end{figure}
Figure~\ref{fig:minmaxtrace} shows the $\delta$-separated local max/min for the same  trace $\sigma_A$ of 
Figure~\ref{fig:simpleminmaxtrace}. 
Observe that the  $\delta$-separated  max/min (Figure~\ref{fig:minmaxtrace}) 
are a subset of the ``simple'' max/min  (Figure~\ref{fig:simpleminmaxtrace}). 
Furthermore the following property holds: 
\begin{property}
For $\delta\!=\! 1$ the sequence of $\delta$-separated  maxima (minima) 
of an alternating trace $\sigma_A$ coincides  
with the list of local  maxima (minima), that is: $\sigma^{M_1}_A\!=\! \sigma^{M}_A$ and 
$\sigma^{m_1}_A\!=\! \sigma^{m}_A$. 
\end{property}

The detection  of  the $\delta$-separated local maxima (minima) for a trace $\sigma_A$ can be described 
in terms of an iterative procedure through which the list of detected max/min 
are constructed as $\sigma_A$ unfolds. 
Such a procedure is formally implemented by the LHA ${\cal A}_{peaks}$ (Figure~\ref{fig:LHApeaks}) which we 
introduce later on. 
Here, based on the example illustrated in  Figure~\ref{fig:minmaxtrace}, we informally summarise how detection   of $\delta$-separated max/min  works. 
The detection requires storing of the most recent (temporary) $\delta$-separated max (min) into a variable named $x_M$ ($x_m$),  
while once detection of a  $\delta$-separated maximum (minimum) is completed 
the corresponding variable $x_M$ ($x_m$) is copied into a dedicated list, named $Lmax$, resp. $Lmin$ (see Table~\ref{tab:lhavars_peaks}), 
which contains the detected points. 
To understand how detection works let us consider  the trace in Figure~\ref{fig:minmaxtrace}. 
The first element encounterd is the local minimum $m[1]$ which is then stored into  $x_m=m[1]$. 
 As the trace further unfolds the subsequent maxima are ignored as long as  their distance 
 from the temporary minimum   $x_m$  is less than $\delta$, as is the case with $M[1]$. 
 Similarly any local  minimum $m[i]$ that is encountered after that stored in $x_m$ is ignored (e.g., $m[2]$), unless it is smaller than $x_m$,  in which case $x_m$ is updated with the newly found smaller minimum. 
 As $\sigma_A$ unfolding proceeds we find the next local max $M[2]$ which is distant more than $\delta$ from the temporary 
 minimum $x_m$: this means that $x_m$ currently holds an actual $\delta$-distanced minimum hence its value is appended 
 to $Lmin$ and the procedure starts over, in a symmetric fashion, for the detection of 
 the next maximum.

The rational behind the notion of $\delta$-separated max/min is that 
for locating  the actual peaks of a stochastically   oscillating trace 
it is important to be able to distinguish between  the  minimal peaks corresponding to stochastic noise, 
 the actual peaks of oscillation. With the $\delta$-separated max/min characterisation 
 we provide the modeller with a means to establish  an \emph{observational perspective}: 
by choosing a specific value for the $\delta$ parameters the modeller 
establishes how big a level of noise he/she wants to ignore 
when detecting where the oscillation peaks are located. 

In the remainder we introduce the LHA ${\cal A}_{peaks}$ which 
formally implements the detection of the $\delta$-separated peaks of 
alternating traces.

\begin{figure*}[ht]
\centering
\fbox{
\includegraphics[scale=.4]{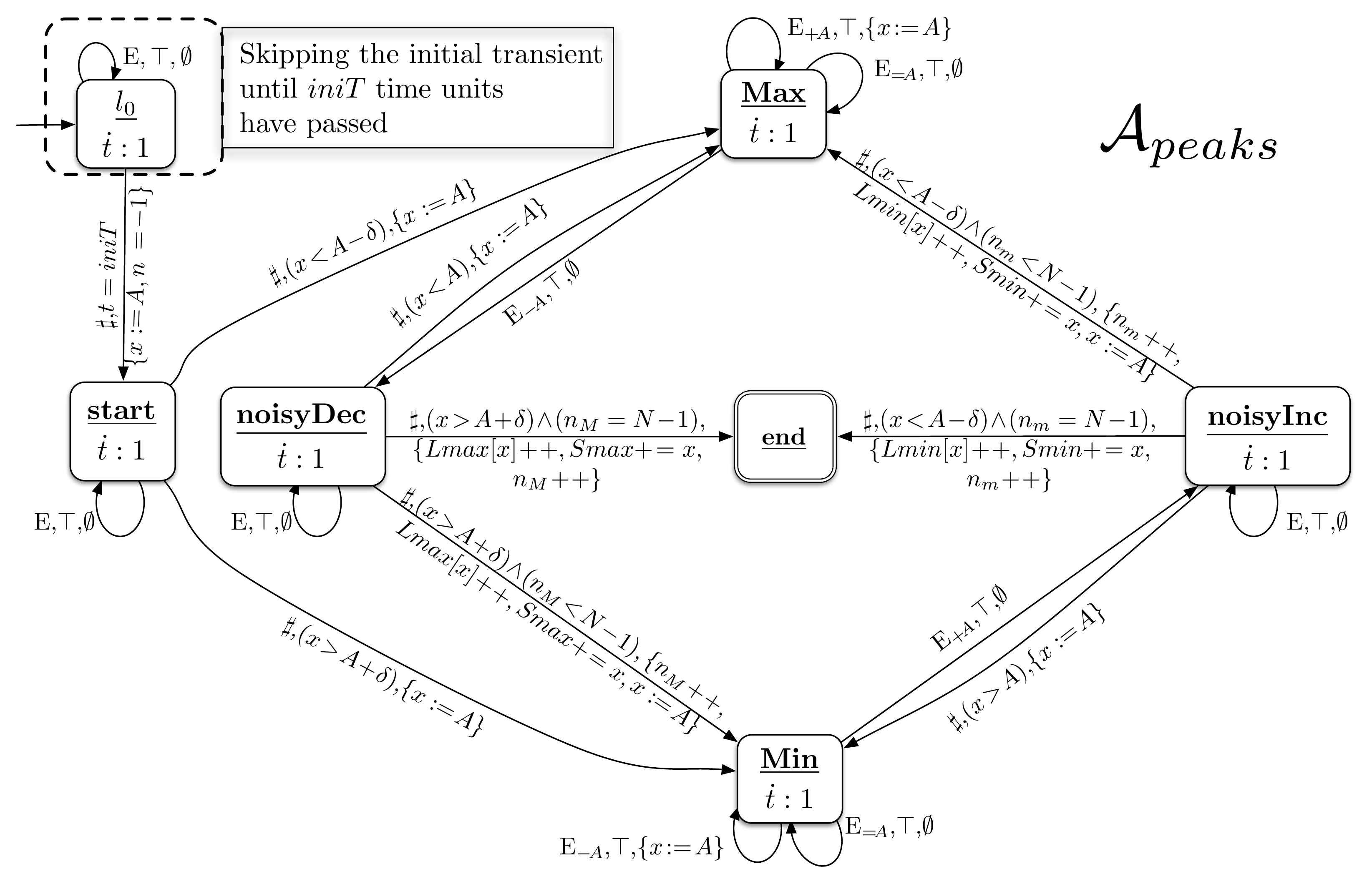}
}
\caption{${\cal A}_{peaks}$: an LHA for detecting local maxima/minima (for observed  species $A$) of  noisy periodic traces where local maxima/minima are detected with respect to a chosen level of noise $\delta$. }
\label{fig:LHApeaks}
\end{figure*}

\ignore{
 \begin{table*}
 \begin{center}
     \begin{tabular}{ | c | c | c | p{4cm} |}
    \hline
    \multicolumn{4}{|c|}{Data variables} \\
    \hline
    name & domain  & update definition & description  \\ \hline
    
        $t$ & $\mathbb{R}_{\geq 0}$ & \emph{reset} & time elapsed since beginning measure (first non-spurious period) \\ \hline
      $n$ & $\mathbb{N}$ & \emph{increment} & counter of  detected local maxima/minima \\ \hline
      $up$ & bool & \emph{complement} & boolean flag indicating whether measuring started with detection of a max or min \\ \hline
    
     $x$ & $\mathbb{N}$ &  \emph{current value of observed species} $A$  & (overloaded) variable storing most recent detected maximum/minumum  \\ \hline
   
    $Smax (Smin)$ & $\mathbb{N}$ &    & sum of detected maxima (minima) \\ \hline
    
    \end{tabular}
\end{center}
    \caption{ The data variables  of automata ${\cal A}_{peaks}$ of Figure~\ref{fig:LHApeaks} for locating the peaks of a noisy oscillatory traces}
    \label{tab:lhavars_peaks}
    \end{table*}
}

 \begin{table*}
 \begin{center}
     \begin{tabular}{ | c | c | p{3cm} | p{4cm} |}
    \hline
    \multicolumn{4}{|c|}{Data variables} \\
    \hline
    name & domain  & update definition & description  \\ \hline
    
        $t$ & $\mathds{R}_{\geq 0}$ & \emph{reset} & time elapsed since beginning measure (first non-spurious period) \\ \hline
      $n_{M} (n_{m})$ & $\mathds{N}$ & \emph{increment} & counter of  detected local maxima ($n_M$), minima ($n_m$) \\ \hline
    
     $x$ & $\mathds{N}$ &  \emph{current value of \emph{observed species} $A$}   & (overloaded) variable storing most recent detected maximum/minumum  \\ 
      \hline
   
    $Smax (Smin)$ & $\mathds{N}$ &    & sum of detected maxima (minima) \\ \hline

$Lmax[] (Lmin[])$& $\mathds{N}^n$ &    & array of frequency of heights of  detected maxima (minima) \\ \hline
    
    \end{tabular}
\end{center}
    \caption{ The data variables  of automaton ${\cal A}_{peaks}$ of Figure~\ref{fig:LHApeaks} for locating the peaks of a noisy oscillatory traces}
    \label{tab:lhavars_peaks}
    \end{table*}

\paragraph{The automaton ${\cal A}_{peaks}$.} 
We introduce an LHA, denoted ${\cal A}_{peaks}$ (Figure~\ref{fig:LHApeaks}),  designed for detecting $\delta$-distanced local maxima/minima along  alternating traces of a given observed species called $A$. 
 It requires a parameter $\delta$ (the chosen noise level) and 
 the partition of the event set $E\!=\!E_{+\!A} \!\cup\! E_{-\!A} \!\cup\! E_{=\!A}$ where $E_{+\! A}$ (respectively $E_{-\! A}, E_{=\! A}$) is the set of events resulting in an increase (respectively decrease, no effect) of the population of $A$. 
 
 The rationale behind the  structure of ${\cal A}_{peaks}$ is to mimic  the cyclic 
structure of an alternating trace through a loop of four locations, two of which (i.e. {\bf\underline{Max}}  and {\bf\underline{Min}})  
are targeted to the detection of local maxima, resp. minima. The simulated trace yields the automaton to 
loop between  {\bf\underline{Max}}  and {\bf\underline{Min}}  hence registering 
the minima/maxima  while doing so. 
The detailed behavior of ${\cal A}_{peaks}$ is as follows. 
Processing of a trace starts with  a configurable  filter of the initial transient 
(represented as a box in Figure~\ref{fig:LHApeaks}) through which  a simulated trace is simply let unfolding for a given $initT$ duration.
The actual analysis begins in location {\bf\underline{start}} 
from which we move to either  {\bf\underline{Max}}  or {\bf\underline{Min}} depending whether we initially 
observe an increase (i.e.  $x<A\!-\!\delta$) or a decrease  (i.e.   $x>A\!+\!\delta$) of the  population of the observed species $A$ 
beyond the chosen level of noise $\delta$. 
Once  within the {\bf\underline{Max}}$\to${\bf\underline{noisyDec}} $\to${\bf\underline{Min}}$\to$ {\bf\underline{noisyInc}} loop the detection 
of local maxima and minima begins. Location  {\bf\underline{Max}} ({\bf\underline{Min}}) is entered  from 
{\bf\underline{noisyInc}}  ({\bf\underline{noisyDec}}) each time 
a sufficiently large (w.r.t. $\delta$) increment (decrement) of $A$ is observed. On  entering  {\bf\underline{Max}} ({\bf\underline{Min}}),
we are sure that the current value of $A$ has moved up (down) of at
least $\delta$ from  the last value stored in $x$ while in
{\bf\underline{Min}} ({\bf\underline{Max}}), hence that value ($x$) is
an actual local minimum (maximum) thus we add it up to $Smin$
($Smax$), then we increment the frequency counter corresponding to the
level of the detected minimum $Lmin[x]$ (maximum $Lmax[x]$)\footnote{with a slight abuse of notation 
we refer to $Lmin[]$ and $Lmax[]$ as arrays whereas in reality within COSMOS/HASL they 
correspond to a set of variables $Lmin_i$, $Lmax_j$, each of which is associated to a given level of the observed population, 
thus $Lmin_1$ counts the frequency of observed minimum at value 1, $Lmin_2$ the observed minima 
at value 2 and so on. 
The number of required $Lmin_i$, $Lmax_j$ variables, which is potentially infinite, can be 
actually bounded without loss of precision to a sufficiently large value $Lmin_m$ ({\em resp.} $Lmax_m$) which must be established manually  
beforehand, for example by observing few previously generated traces.  
}
  before storing the new value of $A$ in $x$ 
and finally increase $n_M$ ($n_m$) the counter of detected maxima (minima). 
Once in {\bf\underline{Max}} ({\bf\underline{Min}}) we stay there as long as we observe the occurrence of 
reactions which do not decrease (increase) the value of $A$, hence either a  reaction of  $E_{+\!A}$ ($E_{-\!A}$), in which case we also 
store the new increased (decreased) value of $A$, hence a potential next local maximum (minimum)  in $x$, or one of $E_{=\!A}$.  
On the other hand on occurrence of a ``decreasing'' (``increasing'') reaction $E_{-\!A}$ ($E_{+\!A}$) 
we move to  {\bf\underline{noisyDec}} ({\bf\underline{noisyInc}}) from which we can either move back to {\bf\underline{Max}} ({\bf\underline{Min}}), 
if we observe a new increase (decrease) that makes the population of
$A$ overpass $x$ ($x$ overpass $A$), or eventually entering  {\bf\underline{Min}} ({\bf\underline{Max}}) 
as soon as the observed decrease (increase) goes beyond the chosen $\delta$ (see above). 
For the automaton ${\cal A}_{peaks}$ depicted in Figure~\ref{fig:LHApeaks}, the  analysis of the simulated trace ends, 
by entering the   {\bf\underline{end}} location either from {\bf\underline{noisyDec}} or {\bf\underline{noisyInc}}, 
as soon as $N$ maxima (or  minima, depending on whether the first observed  peak was a maximum or a minimum) have been detected. 
Notice that ${\cal A}_{peaks}$ can straightforwardly be adapted to different ending conditions. 
The data variables of ${\cal A}_{peaks}$ are summarised  in Table~\ref{tab:lhavars_peaks}. \\

\paragraph{HASL expressions associated to ${\cal A}_{peaks}$.}
We  define different HASL expressions to be associated to to automaton ${\cal A}_{peaks}$. 

\begin{itemize}
\item $Z_{max}\equiv E[last(Smax)/n_M]$: corresponding to  the expected
  value of the average height of the maximal peaks for the first $N$  detected maxima. 
\item $Z_{min}\equiv E[last(Smin)/n_m]$: same as $Z_{max}$ but for minima. 
\item$Z_{PDFmax}\equiv E(last(Lmax)/n_M)$:
    enabling to compute 
  the PDF of  the  height (along a path) of the maximal peaks   
\item $Z_{PDFmin}\equiv E(last(Lmin)/n_m)$:  enabling to
compute the PDF of  the  height (along a path) of the maximal peaks   
\end{itemize} 


\noindent
Expression $Z_{max}$ ($Z_{min}$)  represents  the 
average value of the detected $\delta$-separated maxima (minima). This is obtained by considering  the sum of all detected $\delta$-separated local maxima (minima), which is stored in $Smax$ ($Smin$) and dividing it  by the number of detected maxima $n_m$ ($n_m$). 
Expression $Z_{PDFmax}$ ($Z_{PDFmin}$) allows to estimate the PDF of the height of the detected $\delta$-separated local maxima (minima). 
This is achieved by dividing  the frequency counters of each detected maximal (minimal) peak's height, whose values are  stored in  array $Lmax$ ($Lmin$),    by $n_M$ ($n_m$), the number of detected 
maxima  (minima). 
\ignore{
\paragraph{Discussion.} A limitation of the methodology  based on \emph{noisy periodicity} described in the previous section 
is that the resulting measures are affected by the chosen  partition of the  domain of the observed species  (i.e. the  $L<H$ parameters 
of  ${\cal A}_{per}$). 
The effect that  parameters $L\!<\!H$ have on the measure  depends not only on the distance $H\!-\! L$ but also on the location of 
$L$ and $H$ with respect to the ``centre" of oscillation.
Simply speaking ${\cal A}_{per}$ can be seen as a filter for measuring oscillating traces of given minimal amplitude $H\!-\! L$ and whose peaks 
are located above $H$ and below $L$. 
Thus through ${\cal A}_{per}$ we leave the user the responsibility to decide the minimal amplitude (and peaks location) of oscillatory traces 
whose period  he/she wants to measure. 
}

\section{Case study}
\label{sec:circlock}
\noindent
To demonstrate the above described procedure we consider  a popular example of oscillator, the so-called 
circadian clock.  Circadian clocks are biological mechanisms responsible  for keeping track of daily cycles of light and darkness. 
Here we focus on a model of the biochemical network  (\cite{VKBL02}) 
which is believed to be at the basis of the control of circadian clocks. The network (Figure~\ref{fig:circClock}) involves 2 genes, $D_A$ which expresses the {\it activator} protein $A$ 
\begin{figure}[htbp]
\begin{center}
\includegraphics[width=0.65\linewidth]{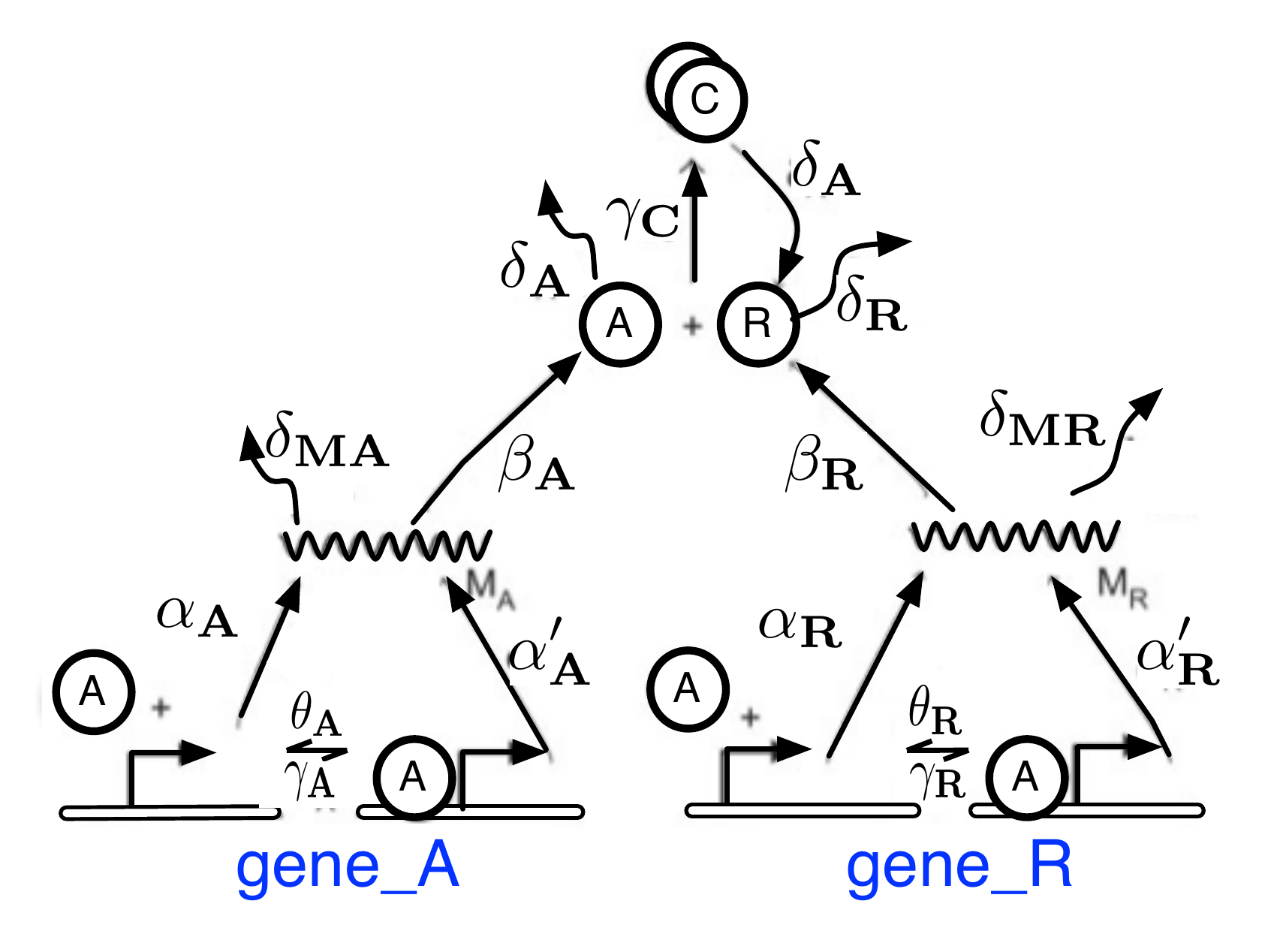}
\caption{Circadian Clock oscillator network: gene $D_A$ expresses activator protein $A$ through transcription of mRNA $M_A$, while gene $D_R$ expresses the repressor protein $R$ 
through transcription of mRNA $M_R$. 
}
\label{fig:circClock}
\end{center}
\end{figure}
and $D_R$ which expresses the {\it repressor} protein $B$.  Protein expression is a two steps process: 
in the first phase a gene transcribes a messenger RNA (mRNA) molecule; in the second phase 
the mRNA molecule is translated into the target protein. For the model of circadian clock we consider here we denote  $M_A$, 
the mRNA species transcribed by gene $D_A$, and   $M_R$ the mRNA transcribed by gene $D_R$. $M_A$ and $M_R$ are then translated into proteins $A$, respectively $R$ .

\begin{small}
\begin{equation}
\label{eq:circclock}
\begin{aligned}
R_1: & A + D_A  \stackrel{\gamma_A}{\longrightarrow}   D'_A &\hspace{.5cm} & R_9:  M_A   \stackrel{\beta_A}{\longrightarrow}   M_A+ A  \\
R_2: & D'_A  \stackrel{\theta_a}{\longrightarrow}   A + D_A &\quad &  R_{10}:  M_R   \stackrel{\beta_R}{\longrightarrow}   M_R+ R \\
R_3: & A + D_R  \stackrel{\gamma_R}{\longrightarrow}  D'_R &\quad &  R_{11}: A+R   \stackrel{\gamma_C}{\longrightarrow} C\\
R_4: & D'_R  \stackrel{\theta_R}{\longrightarrow}  D_R + A &\quad &  R_{12}: C   \stackrel{\delta_A}{\longrightarrow} R \\
R_5: & D'_A  \stackrel{\alpha'_A}{\longrightarrow}  M_A + D'_A  &\quad &  R_{13}: A   \stackrel{\delta_A}{\longrightarrow} \emptyset \\
R_6: & D_A  \stackrel{\alpha_A}{\longrightarrow}  M_A + D_A  &\quad &  R_{14}: R   \stackrel{\delta_R}{\longrightarrow} \emptyset \\
R_7: & D'_R  \stackrel{\alpha'_R}{\longrightarrow}  M_R + D'_R  &\quad &  R_{15}: M_A   \stackrel{\delta_{M_A}}{\longrightarrow} \emptyset \\
R_8: & D_R  \stackrel{\alpha_R}{\longrightarrow}  M_R + D_R  &\quad & R_{16}: M_R   \stackrel{\delta_{M_R}}{\longrightarrow} \emptyset \\
\end{aligned}
\end{equation}
\end{small}    

Protein $A$ acts as an activator for  both genes by attaching to  promoter region of $D_A$ and  $D_R$ (i.e. when $A$ is attached to a gene 
the mRNA transcription increases). Species $D'_A$ and $D'_R$ 
represent the state of gene $D_A$, respectively $D_B$, when an activator molecule ($A$) is attached to their promoter.  
 Note that gene $D_R$ acts as a repressor of  $D_A$ since  
when $A$ bounds to its promoter  $D_R$   sequesters the activator $A$ and, as  a result, the transcription of $D_A$ slows down. 
The repressing role of  $D_R$ is further due to the fact that the expressed protein $R$ inactivates the activator $A$ by binding to it and forming 
the complex $C$. Finally the model in Figure~\ref{fig:circClock} accounts for degradation of all species: thus the mRNAs  $M_A$ and $M_R$, as well as the expressed proteins $A$ and $B$ 
degrades with given rates (see Table~\ref{tab:rates}).  Notice that proteins $A$ degrades also when attached to $R$ (i.e. when in complex $C$), and, as a consequence, $C$ turns into $R$ 
at a rate equivalent to the degradation rate of $A$. 

\begin{table*}
\begin{center}
\begin{small}
\begin{tabular}{|c|c||c|c||c|c||c|c|}
\hline
$\alpha_A$ & $50\ h^{-1}$ & $\alpha_R$ & $0.01\ h^{-1}$ & $\delta_{A}$ & $1\ h^{-1}$ & $\delta_{R}$ & $0.2\ h^{-1}$\\
$\alpha_{A'}$ & $500\ h^{-1}$ & $\alpha_{R'}$ & $50\ h^{-1}$ & $\gamma_{A}=\gamma_{R}$ & $1\ mol^{-1}h^{-1}$ & $\gamma_{C}$ & $2\  mol^{-1}h^{-1}$\\
$\beta_A$ & $50\ h^{-1}$ & $\beta_R$ & $5\ h^{-1}$ & $\theta_A$ & $50\ h^{-1}$ & $\theta_R$ & $100\ h^{-1}$\\
$\delta_{MA}$ & $10\ h^{-1}$ & $\delta_{MR}$ & $0.5\ h^{-1}$ &&& &\\

\hline
\end{tabular}
\end{small}   
\end{center}
\caption{reactions' rates for the circadian oscillator}
\label{tab:rates}
\end{table*}

\ignore{
\begin{figure*}[ht]
\centering
\fbox{
\includegraphics[width=0.8\textwidth]{PIC/circadianClock_GSPN}
}
\caption{GSPN encoding of the system~(\ref{eq:circclock}) of chemical equations corresponding to the circadian-clock network of Figure~\ref{fig:circClock}. Transitions are associated to 
(single-server) marking-dependent exponential distributions whose rates correspond to those in Table~\ref{tab:rates}.}
\label{fig:circClockgspn}
\end{figure*}
}
 \begin{figure*}[ht]
\centering{
\fbox{
\includegraphics[width=0.75\textwidth]{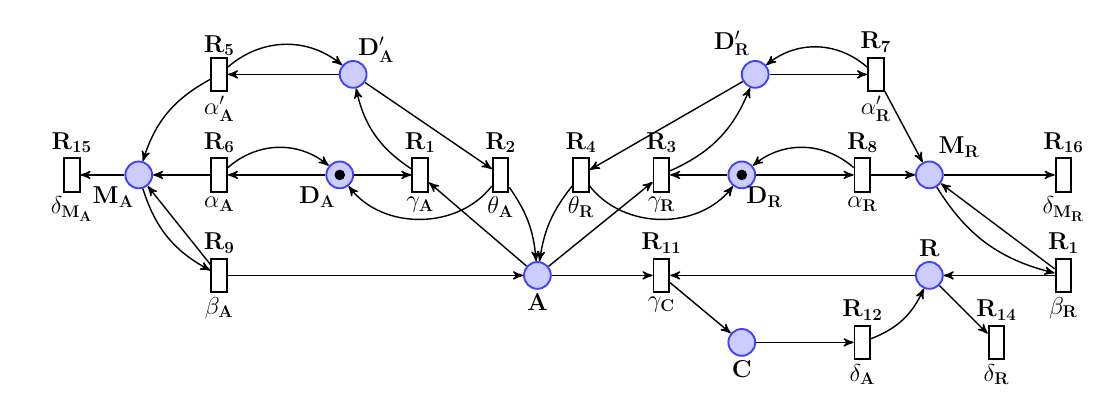}
}
 \caption{GSPN encoding of the system~(\ref{eq:circclock}) of chemical equations corresponding to the circadian-clock.}
 \label{fig:circClockgspn} }
 \end{figure*}

\begin{figure*}[ht]
\centering
\includegraphics[scale=0.22]{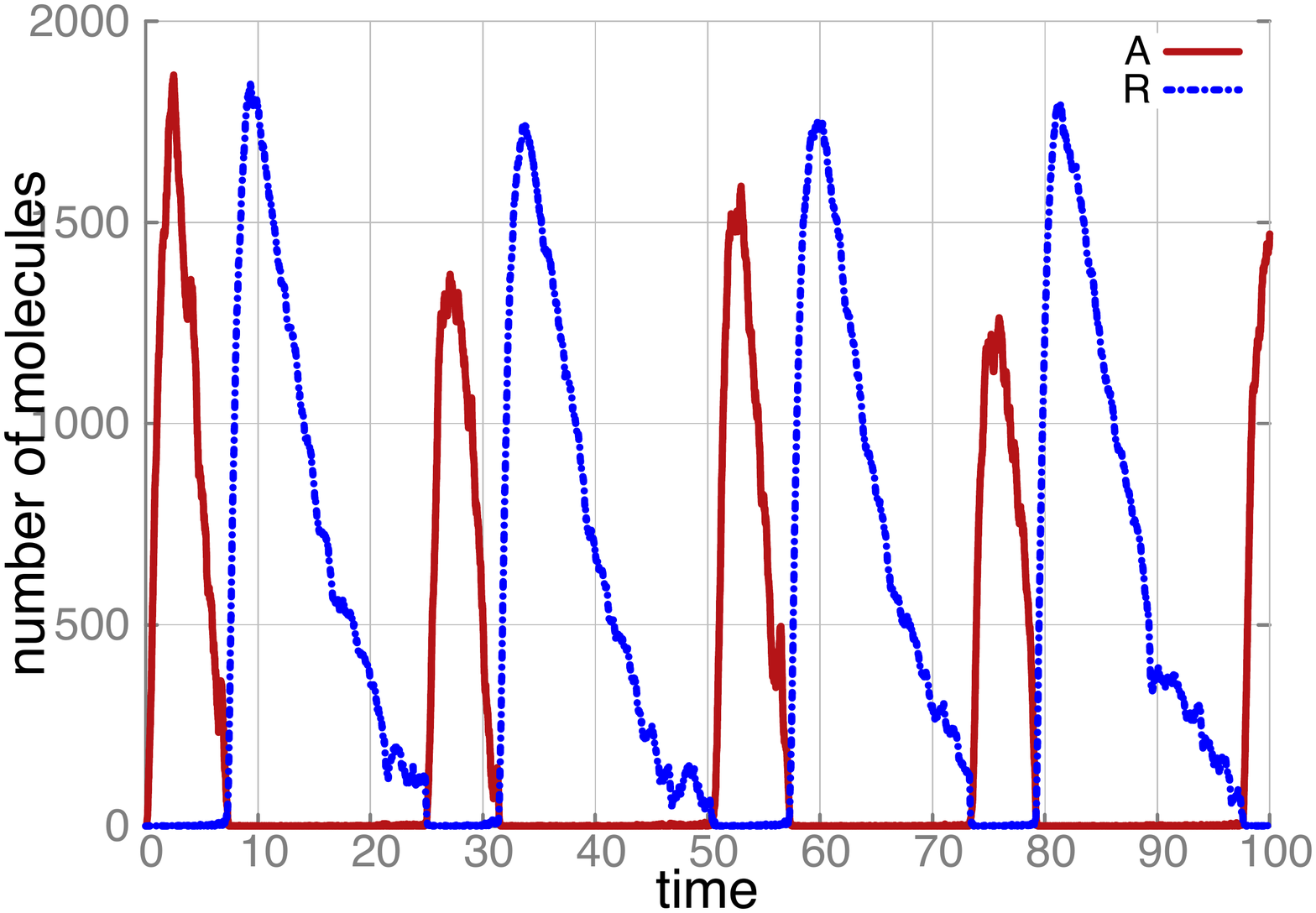}
\hskip 4ex
\includegraphics[scale=0.22]{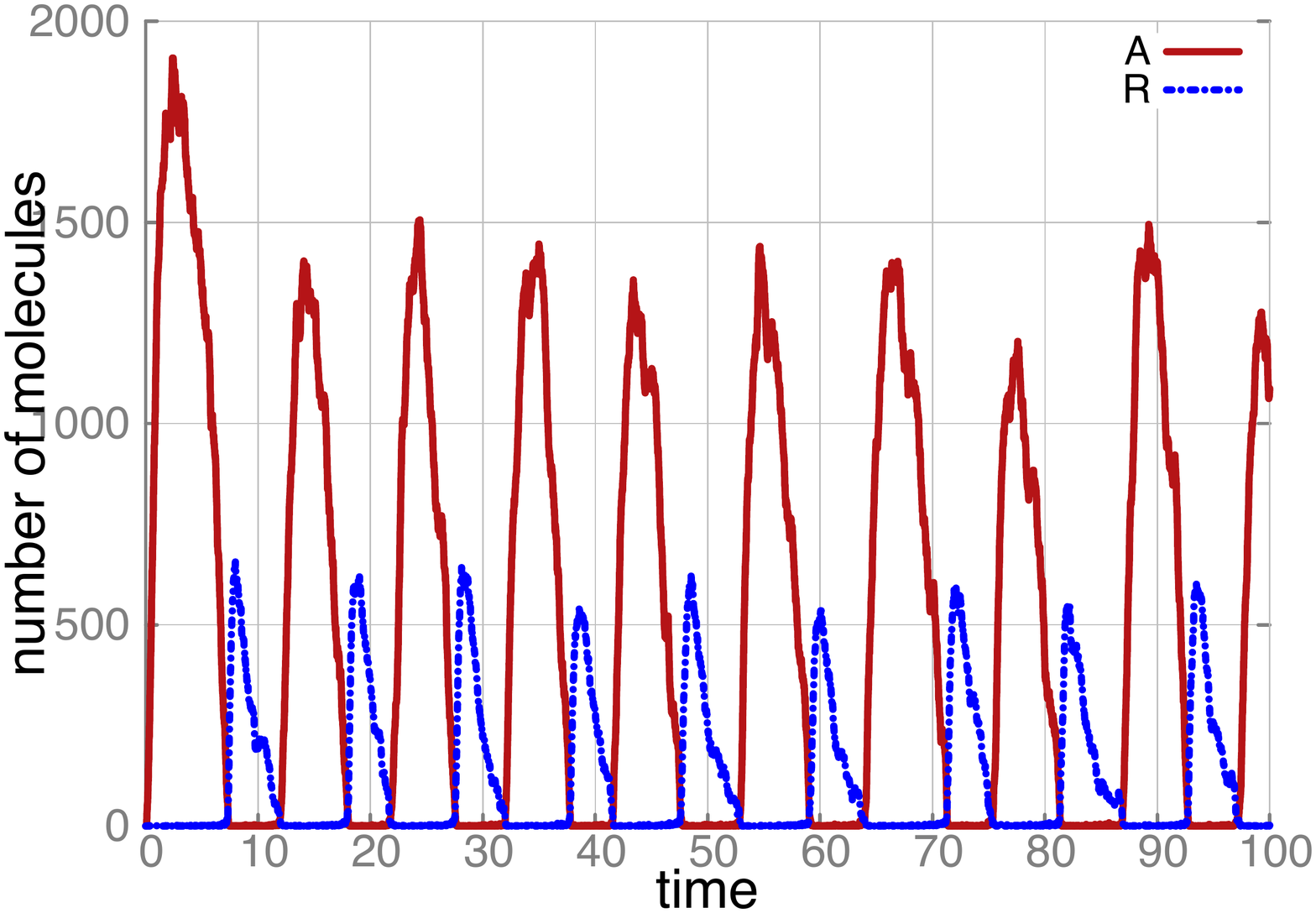}
\caption{Single trajectory showing the oscillatory character of activator $A$ and repressor $B$ dynamics with normal repressor's degradation rate $\delta_R\!=\! 0.2$ (left) and with $10\times$ speed-up, i.e. $\delta_R=2$ (right).  }
\label{fig:scaled_diss}
\end{figure*}

\noindent
The  model of Figure~\ref{fig:circClock} corresponds to the  system of chemical equations~(\ref{eq:circclock}), whose 
(continuous) kinetic rates  (taken from~\cite{VKBL02}) are given in Table~\ref{tab:rates}.

\ignore{
\begin{table}[htdp]
\begin{center}
\begin{small}
\begin{tabular}{|c|c||c|c|}
$\alpha_A$ & $50\ h^{-1}$ & $\alpha_R$ & $0.01\ h^{-1}$\\
$\alpha_{A'}$ & $500\ h^{-1}$ & $\alpha_{R'}$ & $50\ h^{-1}$\\
$\beta_A$ & $50\ h^{-1}$ & $\beta_R$ & $5\ h^{-1}$\\
$\delta_{MA}$ & $10\ h^{-1}$ & $\delta_{MR}$ & $0.5\ h^{-1}$\\
$\delta_{A}$ & $1\ h^{-1}$ & $\delta_{R}$ & $0.2\ h^{-1}$\\
$\gamma_{A}=\gamma_{R}$ & $1\ mol^{-1}h^{-1}$ & $\gamma_{C}$ & $2\  mol^{-1}h^{-1}$\\
$\theta_A$ & $50\ h^{-1}$ & $\theta_R$ & $100\ h^{-1}$\\
\end{tabular}
\end{small}   
\end{center}
\caption{The rates of the circadian clock reactions~(\ref{eq:circclock})} 
\label{tab:rates}
\end{table}
}

\begin{figure*}[ht]
\centering
\includegraphics[scale=0.25]{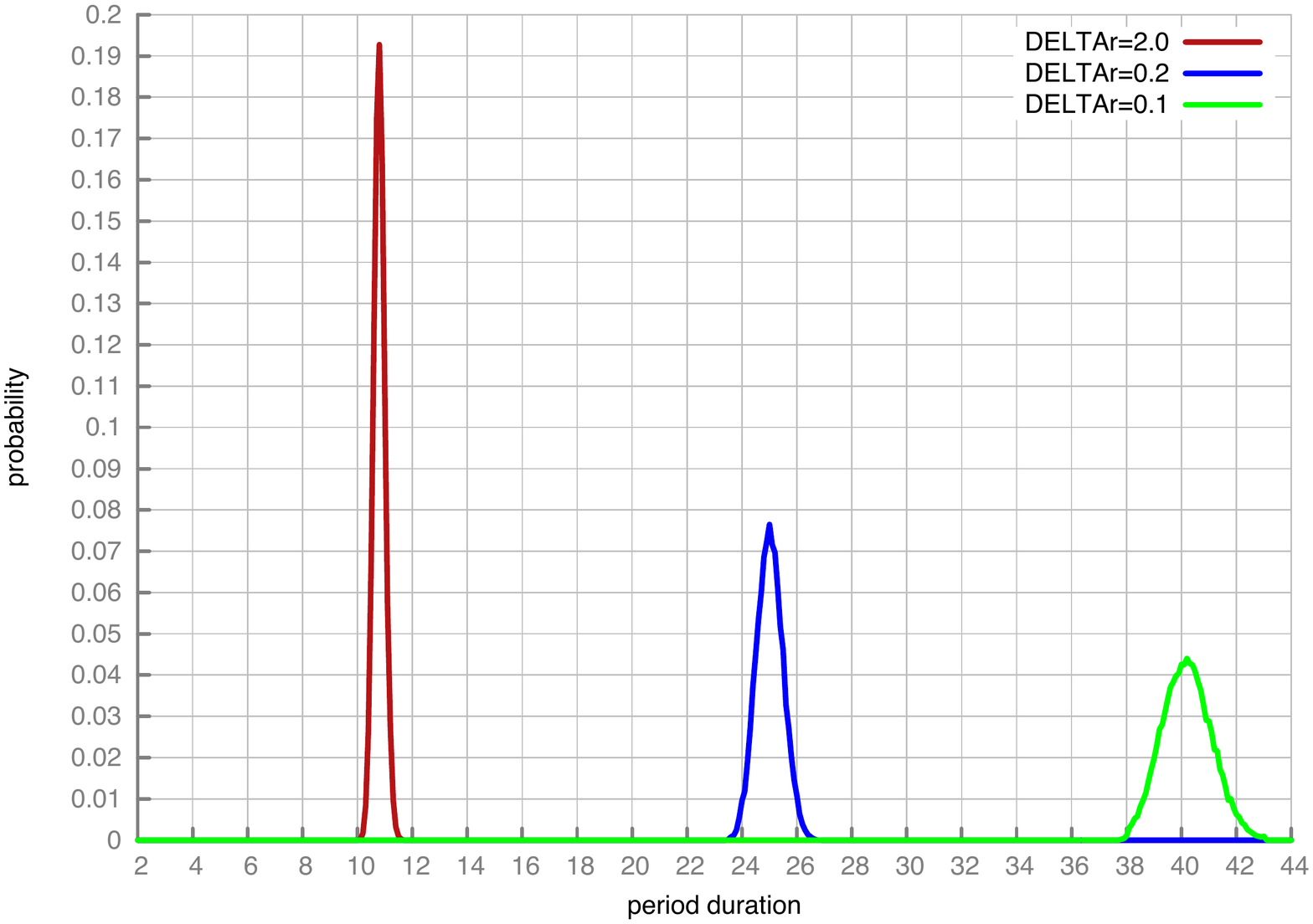}
\includegraphics[scale=0.25]{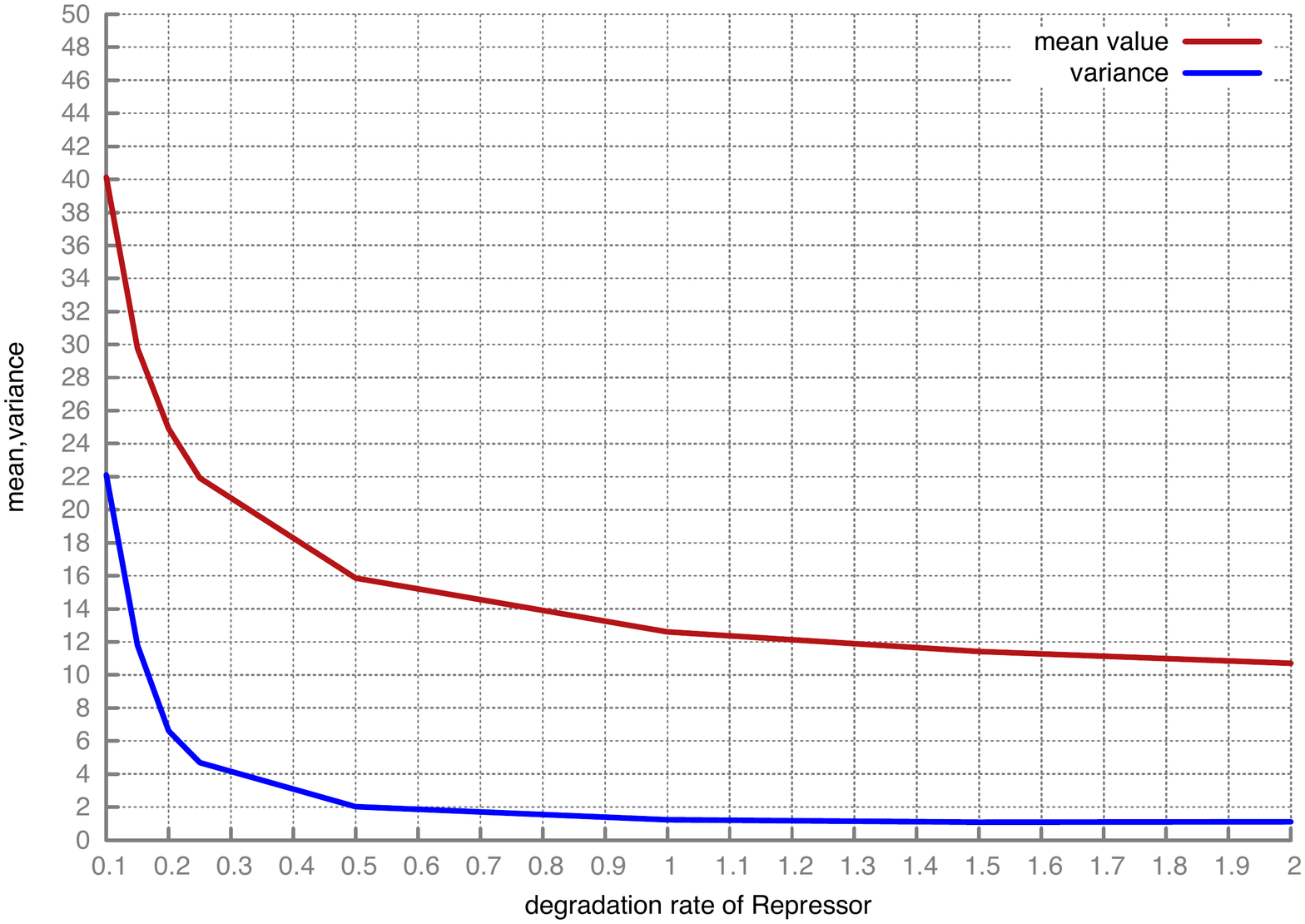}
\caption{The PDF   (left) and the mean value vs the fluctuation (right) of the period of oscillations of protein $A$  of the circadian clock measured with ${\cal A}_{period}$ 
in function of the repressor's degradation rate.}
\label{fig:pdfperiod}
\end{figure*}

\paragraph{Stochastic model} Equations~(\ref{eq:circclock}) can give rise to either  a system of ODEs 
or  to a stochastic process. 
Here we focus on the discrete-stochastic semantics: Figure~\ref{fig:circClockgspn} 
shows the GSPN encoding of equations~(\ref{eq:circclock}) developed  with \cosmos. 
The configuration of the GSPN (i.e. the stochastic process) requires setting the initial population and the rates of each transition (i.e. reaction). 
For the initial population, following~\cite{VKBL02}, we observe that the model comprises one gene $D_A$ and one $D_R$, which can either be in free-state (no activator $A$ is attached 
to the promoter) or in activator-bound state, i.e. $D'_A$, respectively $D'_R$. As a consequence the population of species $D_A$ and $D_R$ is 
bounded by the following invariant constraints: $D_A+D'_A=1$ and $D_R+D'_R=1$ (in fact places DA, DA' and DR, DR' of net  in Figure~\ref{fig:circClockgspn} are the only places covered by P-invarriants). The remaining species are initially supposed to be ``empty", hence they are initialised to $0$. 
Concerning  the transition rates, for simplicity we assume a unitary volume of the system under consideration, hence all continuous rates in Table~\ref{tab:rates} can be used straightforwardly as 
rates of the corresponding discrete-stochastic reactions. In this case we assume all reactions following a negative exponential law.

The oscillatory dynamics of the GSPN model of Figure~\ref{fig:circClockgspn} can be observed by plotting of a simulated trajectory (Figure~\ref{fig:scaled_diss}). 
Observe that the frequency of oscillations varies  considerably with the degradation rate of the repressor ($R$) protein:  a faster degradation 
of $R$ (right), intuitively, results in a higher frequency of oscillations. 
In the remainder we formally assess the oscillatory characteristics (i.e. the period and the peaks of oscillations) of the 
circadian clock model by application of  the previously described approach, i.e. by 
analysing the stochastic process deriving from synchronisation of the circadian clock GSPN model with the ${\cal A}_{period}$ and ${\cal A}_{peaks}$ automata. 

\paragraph{Measuring   the period of the circadian clock.}
We performed a number of experiments aimed at assessing the effect that the degradation rate of the repressor protein ($\delta_R$) has 
on the period of the circadian oscillator. 
Figure~\ref{fig:pdfperiod} (right) shows three plots representing the  PDF of the period (obtained through 
the HASL formula $({\cal A}_{period}, PDF(Last(t)/N))$ for three values of $\delta_R$.  With $\delta_R=0.2$ (i.e. the original value as given 
in~\cite{VKBL02}) the PDF is centred at $t=24.9$, i.e. slightly more of the standard 24 hours period expected for a circadian clock. On the other hand speeding 
up the   repressor degradation of 10 times (i.e. $\delta_R=2$) yields a  slightly more than halved  oscillation period  (i.e. PDF centred at $T=10.8$). Finally slowing down the 
degradation rate of a half (i.e. $\delta_R=0.1$) yields a less than doubled oscillation period (i.e. PDF centred at $T=40.7$). 

\begin{figure*}[ht]
\centering
\includegraphics[scale=.25]{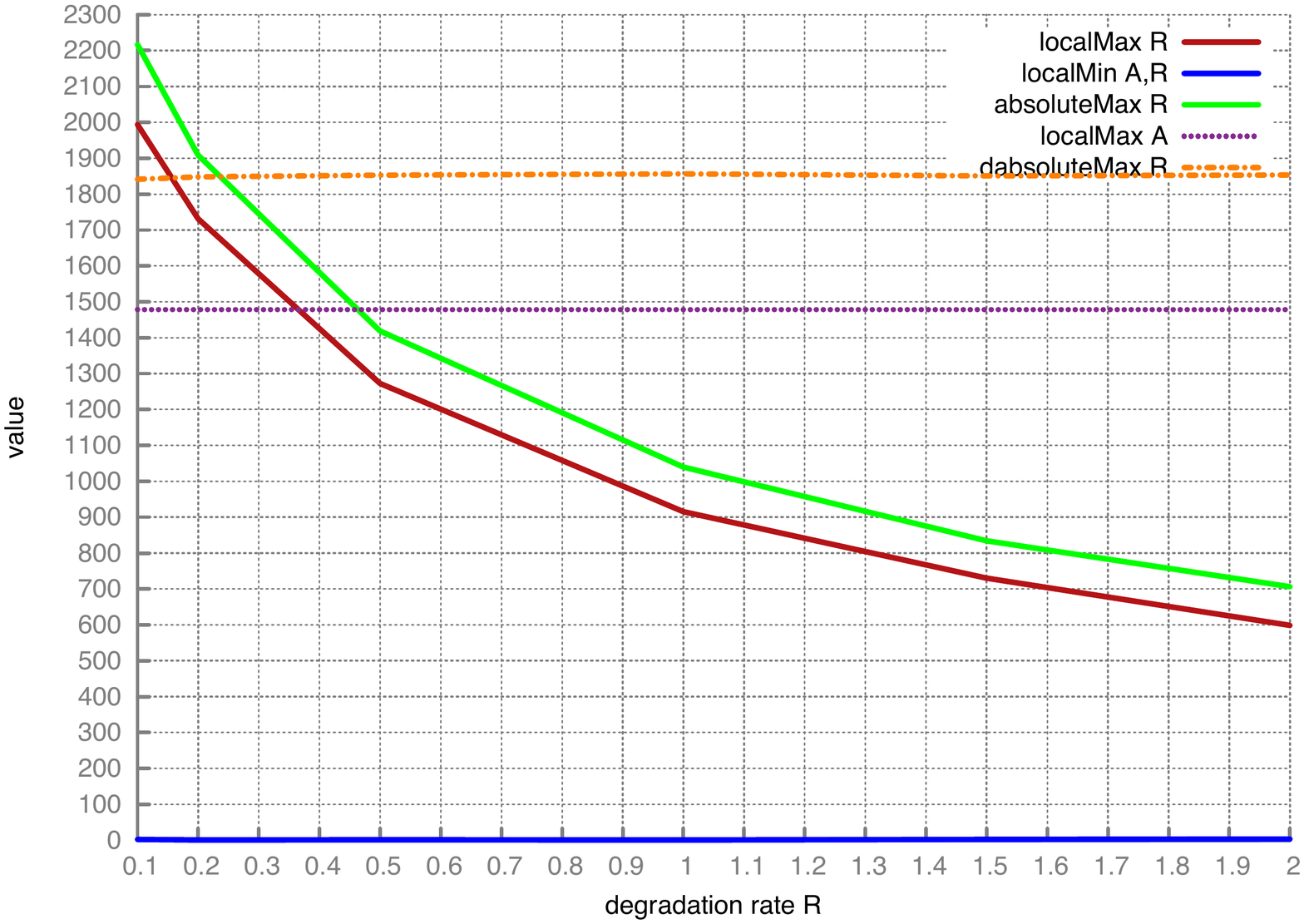}
\caption{The mean value of the minimal and maximal peaks of proteins $A$ and $R$ of the circadian clock measured with ${\cal A}_{peaks}$ 
in function of the repressor's degradation rate.}
\label{fig:meanPeaks}
\end{figure*}

Figure~\ref{fig:pdfperiod} (left) shows plots for the period mean value (red plot) and the period fluctuation (blue plot, as described in Section~\ref{sec:periodHASL}) in function of 
the degradation rate $\delta_R$. They   indicate that slowing down the  degradation of the repressor yields, on one hand,  to a lower the frequency of oscillations, and on the other,  augmenting the irregularity of the periods (i.e. augmenting the period's fluctuations). 
All plots in Figure~\ref{fig:pdfperiod} result from sampling of finite trajectories consisting of $N\!=\!100$ periods, where periods have been detected using  $L\!=\!1$ and $H\!=\!1000$ 
as partition thresholds, and target estimates have been computed with confidence level 99 and confidence-interval width of 0.01. Furthermore the PDF plots in Figure~\ref{fig:pdfperiod} (right) have been 
computed using  a discretisation of the period support interval $[0,50]$ into subintervals of width 0.1.

\paragraph{Measuring   the peaks of oscillations  of the circadian clock.}
We performed a number of experiments aimed at assessing the effect that the degradation rate of the repressor protein ($\delta_R$) has 
on the peaks of oscillation for  both protein $A$ and $R$. 
Figure~\ref{fig:meanPeaks} shows  plots for the mean value of  the minimal and maximal peaks of 
oscillations for both $A$ and $R$. 
Results indicate that while the 
degradation rate $\delta_R$ has no effect on the oscillation peaks of $A$ (both maximal and minimal peaks of $A$ are constant independently of $\delta_R$), 
it affects the maximal peaks (only) of $R$. Specifically  the mean value of $R$'s maximal peaks decreases with the increasing of $\delta_R$ 
(while the minimal peaks of $R$   are constantly at $0$), notice that this is in agreement with 
what indicated by the  single trajectories depicted in  Figure~\ref{fig:scaled_diss}. 
Notice that Figure~\ref{fig:meanPeaks} contains also plot for the absolute maximum of population of  $A$ and $R$ 
measured along the sampled trajectories through trivial  HASL expressions $Z\equiv AVG(max(x))$ (where $x$ is a variable 
used to record the population of the observed species along a synchronising path). 
All plots in Figure~\ref{fig:meanPeaks} result from sampling of finite trajectories containing of $N\!=\!100$ maximal peaks and using a noise parameter $\delta\!=\! 10\% AVG(max(x))$, 
meaning that for evaluating the mean value of maximal peaks we discarded all  critical points distanced one another less than 10\% of the absolute maximum of the observed species.  
Finally, again points of every plot in Figure~\ref{fig:meanPeaks}   have been computed with confidence level 99 and confidence-interval width of 0.01.

\paragraph{On the initial transient.}
To assess the effect that the initial transient of the circadian clock model 
 have on the period and peaks estimates we repeated all of the above discussed 
experiments with different values  of the $initT$ parameter (e.g. $initT\!\in\!\{10,50,100,500,1000,\ldots\}$) 
which determines the starting measuring point for  ${\cal A}_{period}$ 
and of ${\cal A}_{peaks}$. The outcomes of repeated experiments   turned out to be     independent of the 
chosen $initT$ value, indicating that circadian clock reaches its steady state very quickly. 
\section{Related work and discussion}
The HASL based  methodology  presented in this paper is by no means the only  
approach aimed at the analysis of discrete-state stochastic oscillators. 
In the following we provide a brief (non exhaustive) overview of similar approaches. 

\paragraph{Mathematical approaches} 

The analysis of   periodic signals can be achieved through well established signal processing techniques such as, for example,  Fast Fourier Transform (FFT) and autocorrelation. 
Both methods   estimate  the dominant frequency of a periodic signal   given  in terms of a sequence of (real-valued) 
points. 
In the context of stochastic modelling both  FFT and autocorrelation analysis is performed 
over trajectories   generated by a stochastic simulator. 
In order to increase the accuracy of the estimates   usually  frequency estimation is then replicated  over $N$  trajectories,  
the final result being given as the  average of the frequency estimate of each trajectory (see e.g.~\cite{DBLP:conf/isola/DavidLLMPS12,ihekwaba:hal-00784412}).
The main appeal   of signal processing techniques  is due to their   simplicity. However, in the context of statistical model checking, 
 adding an (automatic)  control on the accuracy 
of the resulting estimate would require their integration within a confidence interval estimation procedure, something which 
at best of our knowledge has not yet been done. 
From an expressiveness point of view it is worth remarking that  FFT and autocorrelation are limited to 
estimating the (mean value) of the frequency of an oscillator but provides no support for assessing other aspects 
of oscillator such as the location of the oscillation peaks and the \emph{regularity} (i.e. the fluctuation)  of  the period. 
Finally another interesting contribution belonging to the field of mathematical approaches is presented in~\cite{JKNP12}, where the relationship between stochastic oscillators and their continuous-deterministic counterpart 
is analysed. 




\paragraph{Model checking based approaches.} 
Analysis of oscillators through stochastic model checking techniques has been considered in several works.  
Application of CSL~\cite{BHHK03}  to the characterisation  CTMC biochemical oscillators  has been considered, 
with limited success in~\cite{Ballarini20102019},  and more comprehensively in~\cite{Spieler09,Spieler13}. 
In~\cite{Spieler09} Spieler demonstrated that deciding whether a given CTMC  model oscillates sustainably  
boils down to a steady-state analysis problem where the allegedly oscillating CTMC is coupled with 
a \emph{period detector} automata (through  manual hard-wiring). In this case the 
probability that the period of oscillation has a certain value is  computed through  dedicated CSL steady-state formulae and has been demonstrated 
through examples on the \prism\ model-checker.

In a recent work  measuring of oscillations has been considered with other statistical model checking tools (UPPAAL-SMC and PLASMA) by application 
of the MITL logic~\cite{DBLP:conf/isola/DavidLLMPS12}. In this case the analysis of period duration is achieved by detection of a single period of oscillation 
through nested time-bounded \emph{Until}  formulae. 

\section{Conclusion}
\label{sec:conc}
We have presented a methodology for the formal analysis of  stochastic models 
exhibiting an oscillatory behaviour. 
Such methodology relies on the application of the HASL formalism, a  statistical model checking 
framework suitable for expressing   sophisticated performance measures. 
We have shown how by means of HASL one can  define   specific LHA automata 
 targeted to the analysis of particular  aspects of the dynamics of oscillatory trajectories, such as: 
 the detection of the  period and of the peaks of a stochastic oscillator.  
For the period we have introduced a class of LHA, denoted ${\cal A}_{per}$,  
for  estimating the PDF, the mean value   as well as the 
 the fluctuation of the period duration (the latter being an interesting measure 
related to the regularity of an oscillator frequency). Concerning the peaks 
we have introduced a class of LHA, ${\cal A}_{peaks}$,  for measuring the mean value of the maximal/minimal peaks of oscillation. 
We have demonstrated the effectiveness of the methodology  by studying a well established model of 
biological oscillator, namely the circadian clock.

\bibliographystyle{abbrv}
\bibliography{biblio}

\end{document}